\pdfoutput=1
\documentclass[acmsmall,screen]{acmart}
\usepackage[utf8]{inputenc} 
\usepackage{amsmath, amsthm}
\usepackage[capitalise]{cleveref}
\usepackage{thmtools}
\usepackage[english]{babel}
\usepackage{paralist}
\usepackage{graphicx}
\usepackage{subcaption}
\usepackage{booktabs}

\usepackage[ruled, linesnumbered]{algorithm2e}
\usepackage{tikz}
\usepackage{parskip}
\usetikzlibrary{arrows,automata,shapes,decorations,decorations.markings,calc, matrix,decorations.pathmorphing}

\crefname{figure}{Figure}{Figure}
\crefformat{footnote}{#2\footnotemark[#1]#3}

\newcommand{\ov}{\overline}
\newcommand{\Nats}{\mathbb{N}}

\newcommand{\Level}{\mathsf{Lv}}
\newcommand{\OpenParenthesis}{\alpha}
\newcommand{\CloseParenthesis}{\ov{\alpha}}
\newcommand{\OpenNonTerminal}{\mathcal{A}}
\newcommand{\CloseNonTerminal}{\ov{\mathcal{A}}}
\newcommand{\StartNonTerminal}{\mathcal{S}}
\newcommand{\Label}{\lambda}
\newcommand{\Path}{\rightsquigarrow}
\newcommand{\Dyck}{\mathcal{L}}
\newcommand{\ParSet}{\Sigma}
\newcommand{\Alphabet}{\Sigma}
\newcommand{\Grammar}{\mathcal{G}}
\newcommand{\BidirectedAlgo}{\mathsf{BidirectedReach}}

\newcommand{\Edges}{\mathsf{Edges}}
\newcommand{\Queue}{\mathcal{Q}}
\newcommand{\SetClosedParenthesis}{\Sigma^C}
\newcommand{\SetOpenParenthesis}{\Sigma^O}
\newcommand{\Produces}{\vdash}
\newcommand{\Tree}{\mathsf{Tree}}
\newcommand{\Bag}{B}
\newcommand{\Shorten}{\mathsf{Shorten}}

\newcommand{\DisjointSets}{\mathsf{DisjointSets}}
\newcommand{\Union}{\mathsf{Union}}
\newcommand{\Find}{\mathsf{Find}}
\newcommand{\MakeSet}{\mathsf{MakeSet}}
\mathchardef\mhyphen="2D

\newcommand{\DataStructure}{\mathcal{D}}
\newcommand{\True}{\mathsf{True}}

\newcommand{\ClosingLabel}{\ov{\Label}}
\newcommand{\Potential}{\Phi}
\newcommand{\restr}[1]{[#1]}
\newcommand{\Comp}{\mathcal{C}}
\newcommand{\Partition}{\mathcal{V}}
\newcommand{\DynamicDyck}{\mathsf{DynamicDyck}}

\newcommand{\Preprocessalgo}{\mathsf{Build}}
\newcommand{\Updatealgo}{\mathsf{Update}}
\newcommand{\Queryalgo}{\mathsf{Query}}
\newcommand{\ReachMap}{R}
\newcommand{\RSMalgo}{\mathsf{Process}}
\newcommand{\Pool}{\mathsf{Pool}}

\newtheorem{remark}{Remark}

\setcopyright{rightsretained} 
\acmPrice{}
\acmDOI{10.1145/3158118}
\acmYear{2018}
\copyrightyear{2018}
\acmJournal{PACMPL}
\acmVolume{2}
\acmNumber{POPL}
\acmArticle{30}
\acmMonth{1}

\begin{document}

\title{Optimal Dyck Reachability for Data-Dependence and Alias Analysis}

\author{Krishnendu Chatterjee}
\affiliation{
  \institution{Institute of Science and Technology, Austria}
  \streetaddress{Am Campus 1}
  \city{Klosterneuburg}
  \postcode{3400}
  \country{Austria}
}
\email{krishnendu.chatterjee@ist.ac.at} 
\author{Bhavya Choudhary}
\affiliation{
  \institution{Indian Institute of Technology, Bombay}
  \streetaddress{IIT Area, Powai}
  \city{Mumbai}
  \postcode{400076}
  \country{India}
}
\email{bhavya@cse.iitb.ac.in}
\author{Andreas Pavlogiannis}
\affiliation{
  \institution{Institute of Science and Technology, Austria}
  \streetaddress{Am Campus 1}
  \city{Klosterneuburg}
  \postcode{3400}
  \country{Austria}
}
\email{pavlogiannis@ist.ac.at}

\begin{abstract}
A fundamental algorithmic problem at the heart of static analysis is Dyck reachability.
The input is a graph where the edges are labeled with different 
types of opening and closing parentheses, and the reachability information is computed 
via paths whose parentheses are properly matched.
We present new results for Dyck reachability problems with applications to 
alias analysis and data-dependence analysis.
Our main contributions, that include improved upper bounds as well as lower bounds that 
establish optimality guarantees, are as follows.

First, we consider Dyck reachability on bidirected graphs,
which is the standard way of performing field-sensitive points-to analysis.
Given a bidirected graph with $n$ nodes and $m$ edges, we present:
(i)~an algorithm with worst-case running time $O(m + n \cdot \alpha(n))$,
where $\alpha(n)$ is the inverse Ackermann function, improving the previously
known $O(n^2)$ time bound;
(ii)~a matching lower bound that shows that our algorithm is optimal
wrt to worst-case complexity;
and (iii)~an optimal average-case upper bound of $O(m)$ time, improving the
previously known $O(m \cdot \log n)$ bound.

Second, we consider the problem of context-sensitive data-dependence analysis,
where the task is to obtain analysis summaries of library code in the presence 
of callbacks.
Our algorithm preprocesses libraries in almost linear time,
after which the contribution of the library in the complexity of the client analysis
is only linear, and only wrt the number of call sites.

Third, we prove that combinatorial algorithms for Dyck reachability on general graphs 
with truly sub-cubic bounds cannot be obtained without obtaining sub-cubic 
combinatorial algorithms for Boolean Matrix Multiplication, which is a long-standing open problem.
Thus we establish that the existing combinatorial algorithms for Dyck reachability 
are (conditionally) optimal for general graphs.
We also show that the same hardness holds for graphs of constant treewidth.

Finally, we provide a prototype implementation of our algorithms for 
both alias analysis and data-dependence analysis.
Our experimental evaluation demonstrates that the new algorithms
significantly outperform all existing methods on the two problems,
over real-world benchmarks.
\end{abstract}

\begin{CCSXML}
<ccs2012>
<concept>
<concept_id>10003752.10010124.10010138.10010143</concept_id>
<concept_desc>Theory of computation~Program analysis</concept_desc>
<concept_significance>500</concept_significance>
</concept>
<concept>
<concept_id>10003752.10003809.10003635</concept_id>
<concept_desc>Theory of computation~Graph algorithms analysis</concept_desc>
<concept_significance>500</concept_significance>
</concept>
</ccs2012>
\end{CCSXML}

\ccsdesc[500]{Theory of computation~Program analysis}
\ccsdesc[500]{Theory of computation~Graph algorithms analysis}

\keywords{Data-dependence analysis, CFL reachability, Dyck reachability, Bidirected graphs, treewidth}

\parskip=0.0\baselineskip \advance\parskip by 0pt plus 0pt
\maketitle
\parskip=0.5\baselineskip \advance\parskip by 0pt plus 2pt

\section{Introduction}\label{sec:intro}
In this work we present improved upper bounds, lower bounds, and 
experimental results for algorithmic problems related to Dyck reachability, 
which is a fundamental problem in static analysis. 
We present the problem description, its main applications, 
the existing results, and our contributions.

\noindent{\bf Static analysis and language reachability.}
Static analysis techniques obtain information about programs without running them on specific inputs. 
These techniques explore the program behavior for all possible inputs and all possible executions. 
For non-trivial programs, it is impossible to  explore all the possibilities, and hence various approximations are used.
A standard way to express a plethora of static analysis problems is via {\em language reachability} that generalizes
graph reachability. 
The input consists of an underlying graph with labels on its edges from a fixed alphabet, and a language,
and reachability paths between two nodes must produce strings that belong to the given
language~\cite{Yannakakis90,Reps97}.

\noindent{\bf CFL and Dyck reachability.}
An extremely important case of language reachability in static analysis is {\em CFL} reachability,
where the input language is context-free, which can be used to model, e.g., 
context-sensitivity or field-sensitivity.
The CFL reachability formulation has applications to a very wide range of static analysis problems,
such as interprocedural data-flow analysis~\cite{PLDI29}, slicing~\cite{Reps94}, 
shape analysis~\cite{PLDI26}, impact analysis~\cite{Arnold96}, type-based flow analysis~\cite{Rehof01} and 
alias/points-to analysis~\cite{PLDI31,PLDI32,PLDI33,PLDI36,PLDI37,PLDI40}, etc.
In practice, widely-used large-scale analysis tools, such as Wala~\cite{Wala} and Soot~\cite{Soot,Bodden12}, 
equip CFL reachability techniques to perform such analyses.
In most of the above cases, the languages used to define the problem are those of properly-matched parenthesis, 
which are known as Dyck languages, and form a proper subset of context-free languages.
Thus Dyck reachability is at the heart of many problems in static analysis.

\noindent{\bf Alias analysis.} 
Alias analysis has been one of the major types of static analysis and a subject of 
extensive study~\cite{Sridharan13,PLDI8,PLDI20,PLDI16}.
The task is to decide whether two pointer variables may point to the same object during program execution.
As the problem is computationally expensive~\cite{PLDI17,PLDI24}, 
practically relevant results are obtained via approximations.
One popular way to perform alias analysis is via points-to analysis, where two variables may alias if their points-to sets intersect.
Points-to analysis is typically phrased as a Dyck reachability problem on Symbolic Points-to  Graphs (SPGs), 
which contain information about variables, heap objects and parameter passing due to method 
calls~\cite{PLDI36,PLDI37}.
In alias analysis there is an important distinction between context and field sensitivity, which we describe below.
\begin{compactitem}
\item {\em Context vs field sensitivity.}
Typically, the Dyck parenthesis are used in SPGs to specify two types of constraints.
\emph{Context sensitivity} refers to the requirement that reachability paths must respect the calling context due to method calls and returns.
\emph{Field sensitivity} refers to the requirement that reachability paths must respect field accesses of composite types in Java~\cite{PLDI32,PLDI33,PLDI36,PLDI37}, or references and dereferences of pointers~\cite{PLDI40} in C.
Considering both types of sensitivity makes the problem undecidable~\cite{Reps00}.
Although one recent workaround is approximation algorithms~\cite{Zhang17}, 
the standard approach has been to consider only one type of sensitivity.
Field sensitivity has been reported to produce better results, and being
more scalable~\cite{Lhotak06}.
We focus on context-insensitive, but field-sensitive points-to analysis.
\end{compactitem}

\noindent{\bf Data-dependence analysis.}
Data-dependence analysis aims to identify the def-use chains in a program.
It has many applications, including slicing~\cite{Reps94}, impact analysis~\cite{Arnold96} and bloat detection~\cite{Xu10}.
It is also used in compiler optimizations, where data dependencies are used to infer whether it is safe to reorder or parallelize program statements~\cite{Kuck81}.
Here we focus on the distinction between \emph{library vs client analysis} and the challenge of \emph{callbacks}.
\begin{compactitem}

\item {\em Library vs Client.} Modern-day software is developed in multiple stages and is interrelated.
The vast majority of software development relies on existing libraries and third-party components which are typically huge and complex.
At the same time, the analysis of client code is ineffective if not performed in conjunction with the library code.
These dynamics give rise to the potential of analyzing library code once, in an offline stage, and creating suitable 
analysis summaries that are relevant to client behavior only. The benefit of such a process is two-fold.
First, library code need only be analyzed once, regardless of the number of clients that link to it.
Second, it offers fast client-code analysis, since the expensive cost of analyzing the huge libraries has been spent offline, in an earlier stage.
Data-dependence analysis admits a nice separation between library and client 
code, and has been studied in~\cite{Tang15,Palepu17}.

\item {\em The challenge of callbacks.}
As pointed out recently in~\cite{Tang15}, one major obstacle to effective library summarization is the presence of callbacks.
Callback functions are declared and used by the library, but are implemented by the client.
Since these functions are missing when the library code is analyzed, library summarization is ineffective 
and the whole library needs to be reanalyzed on the client side, when callback functions become available.
\end{compactitem}

\noindent{\bf Algorithmic formulations and existing results.} 
We describe below the key algorithmic problems in the applications mentioned above and the existing results.
We focus on data-dependence and alias analysis via Dyck reachability, which is the standard way for performing such analysis.
Recall that the problem of Dyck reachability takes as input a (directed) graph, where some edges are marked with opening and closing parenthesis, and the task is to compute for every pair of nodes whether there exists a path between them such that the parenthesis along its edges are matched.

\begin{compactenum}
\item \emph{Points-to analysis.}
Context-insensitive, field-sensitive points-to analysis via Dyck reachability is phrased on an SPG $G$ with $n$ nodes and $m$ edges.
Additionally, the graph is \emph{bidirected}, meaning that if $G$ has an edge $(u,v)$ labeled with an opening parenthesis,
then it must also have the edge $(v,u)$ labeled with the corresponding closing parenthesis.
Bidirected graphs are found in most existing works on on-demand alias analysis via Dyck reachability, and their importance has been remarked 
in various works~\cite{Yuan09,Zhang13}.

The best existing algorithms for the problem appear in the recent work of~\cite{Zhang13}, where two algorithms are proposed.
The first has $O(n^2)$ \emph{worst-case} time complexity; and 
the second has $O(m\cdot \log n)$ \emph{average-case} time complexity and $O(m \cdot n \cdot \log n)$ \emph{worst-case} complexity.
Note that for dense graphs $m=\Theta(n^2)$, and the first algorithm has better average-case complexity too.

\item \emph{Library/Client data-dependence analysis.}
The standard algorithmic formulation of context-sensitive data-dependence analysis is via Dyck reachability, 
where the parenthesis are used to properly match method calls and returns in a context-sensitive way~\cite{Reps00,Tang15}.
The algorithmic approach to Library/Client Dyck reachability consists of considering two graphs $G_1$ and $G_2$, for the library and 
client code respectively.
The computation is split into two phases.
In the \emph{preprocessing phase}, the Dyck reachability problem is solved on $G_1$ (using a CFL/Dyck reachability algorithm), 
and some summary information is maintained, 
which is typically in the form of some subgraph $G'_1$ of $G_1$.
In the \emph{query phase}, the Dyck reachability is solved on the combination of the two graphs $G'_1$ and $G_2$.
Let $n_1$, $n_2$ and $n'_1$ be the sizes of $G_1$, $G_2$ and $G'_1$ respectively.
The algorithm spends $O(n_1^3)$ time in the preprocessing phase, and $O((n'_1+n_2)^3)$ time in the query phase.
Hence we have an improvement if $n'_1>>n_1$.

In the presence of callbacks, library summarization via CFL reachability is ineffective, as $n'_1$ can be as large as $n_1$.
To face this challenge, the recent work of~\cite{Tang15} introduced TAL reachability. 
This approach spends $O(n_1^6)$ time on the client code (hence more than the CFL reachability algorithm),
and is able to produce a summary of size $s<n_1$ even in the presence of callbacks.
Afterwards, the client analysis is performed in $O((s+n_2)^6)$ time, and hence the cost due to the library only appears in terms of its summary.

\item \emph{Dyck reachability on general graphs.}
As we have already mentioned, Dyck reachability is a fundamental algorithmic formulation of many types of static analysis.
For general graphs (not necessarily bidirected), the existing algorithms require $O(n^3)$ time, and they essentially 
solve the more general CFL reachability problem~\cite{Yannakakis90}.  
The current best algorithm is due to~\cite{Chaudhuri08}, which utilizes the well-knwon Four Russians' Trick to exhibit 
complexity $O(n^3/\log n)$.
The problem has been shown to be 2NPDA-hard~\cite{Heintze97}, which yields a conditional cubic lower bound in its complexity.
%The combinatorial cubic barrier for CFL parsing~\cite{Lee02} implies the same barrier for CFL reachability~\cite{Reps97}.
%On the other hand, Dyck parsing is linear-time solvable, which leaves a lot of room for truly sub-cubic algorithms for Dyck reachability.
\end{compactenum}

\noindent{\bf Our contributions.} Our main contributions can be characterized in three parts:
(a)~improved upper bounds; (b)~lower bounds with optimality guarantees; and (c)~experimental results.
We present the details of each of them below.

\noindent{\em Improved upper bounds.} Our improved upper bounds are as follows:

\begin{compactenum}
\item For Dyck reachability on bidirected graphs with $n$ nodes and $m$ edges, we present an algorithm with the following bounds:
(a)~The worst-case complexity bound is $O(m + n\cdot \alpha(n))$ time and $O(m)$ space, 
where $\alpha(n)$ is the inverse Ackermann function, improving the previously known $O(n^2)$ time bound. 
Note that $\alpha(n)$ is an extremely slowly growing function, and for all practical purposes, $\alpha(n) \leq 4$, 
and hence practically the worst-case bound of our algorithm is linear.
(b)~The average-case complexity is $O(m)$ improving the previously known $O(m \cdot \log n)$ bound.
See Table~\ref{tab:bidirected_comparison} for a summary.

\item For \emph{Library/Client Dyck reachability} we exploit the fact that the data-dependence graphs that arise in practice have 
special structure, namely they contain components of small treewidth.
We denote by $n_1$ and $n_2$ the size of the library graph and client graph, and  by $k_1$ and $k_2$ the number of call sites in the library graph and client graph, 
respectively.
We present an algorithm that analyzes the library graph in  $O(n_1+ k_1\cdot \log n_1)$ time and $O(n_1)$ space.
Afterwards, the library and client graphs are analyzed together only in $O(n_2 + k_1\cdot \log n_1 + k_2\cdot \log n_2)$ time and $O(n_1+n_2)$ space.
Hence, since typically $n_1>> n_2$ and $n_i>>k_i$, the cost of analyzing the large library occurs only in the preprocessing phase.
When the client code needs to be analyzed, the cost incurred due to the library code is small.
See Table~\ref{tab:library_comparison} for a summary.

\end{compactenum}

\noindent{\em Lower bounds and optimality guarantees.} 
Along with improved upper bounds we present lower bound and conditional lower bound results
that imply optimality guarantees. 
Note that optimal guarantees for graph algorithms are extremely rare, and we show the algorithms 
we present have certain optimality guarantees.

\begin{compactenum}
\item For Dyck reachability on bidirected graphs we present a matching lower bound of $\Omega(m + n \cdot \alpha(n))$ 
for the worst-case time complexity.
Thus we obtain matching lower and upper bounds for the worst-case complexity, and thus our algorithm is 
optimal wrt to worst-case complexity.
Since the average-case complexity of our algorithm is linear, the algorithm is also optimal wrt the average-case complexity.

\item For \emph{Library/Client Dyck reachability} note that $k_1\leq n_1$ and $k_2\leq n_2$. 
Hence our algorithm for analyzing library and client code is almost linear time, and hence optimal wrt 
polynomial improvements.

\item For Dyck reachability on general graphs we present a conditional lower bound. 
In algorithmic study, a standard problem for showing conditional cubic lower bounds
is Boolean Matrix Multiplication (BMM)~\cite{Lee02,HKNS15,WilliamsW10,AbboudW14}.
While fast matrix multiplication algorithms exist (such as Strassen's algorithm~\cite{Strassen69}), 
these algorithms are not ``combinatorial''\footnote{Not combinatorial means algebraic methods~\cite{LeGall14}, 
which are algorithms with large constants. In contrast, combinatorial algorithms are discrete and non-algebraic; for detailed discussion see~\cite{HKNS15}}. 
The standard conjecture (called the BMM conjecture) is that there is no truly sub-cubic\footnote{Truly sub-cubic means polynomial
improvement, in contrast to improvement by logarithmic factors such as $O(n^3/\log n)$} 
combinatorial algorithm  for BMM, which has been widely used in algorithmic studies for obtaining various types of hardness results~\cite{Lee02,HKNS15,WilliamsW10,AbboudW14}.
We show that Dyck reachability on general graphs even for a single pair is BMM hard. More precisely, 
we show that for any $\delta>0$, any algorithm that solves pair Dyck reachability on general graphs
in $O(n^{3-\delta})$ time implies an algorithm that solves BMM in $O(n^{3-\delta/3})$ time.
Since all algorithms for Dyck reachability are combinatorial, it establishes a 
conditional hardness result (under the BMM conjecture) for general Dyck reachability. 
Additionally, we show that the same hardness result holds for Dyck reachability on graphs of constant treewidth.
%%Moreover, our hardness result holds even if we restrict our attention on graphs of constant-treewidth.
Our hardness shows that the existing cubic algorithms are optimal (modulo logarithmic-factor improvements), 
under the BMM conjecture.
Existing work establishes that Dyck reachability is 2NPDA-hard~\cite{Heintze97}, which yields a a conditional lower bound.
Our result shows that Dyck reachability is also BMM-hard, and even on constant-treewidth graphs, and thus strengthens the conditional cubic lower bound for the problem.
\end{compactenum}

\begin{table}
\centering
\renewcommand{\arraystretch}{1.5}
\caption{Comparison of our results with existing work for Dyck reachability on bidirected graphs with $n$ nodes and $m$ edges.
We also prove a matching lower-bound for the worst-case analysis. 
}
\label{tab:bidirected_comparison}
\small
\begin{tabular}{|c||c|c|c|c|c|}
\hline
& \textbf{Worst-case Time} & \textbf{Average-case Time} & \textbf{Space} & \textbf{Reference}\\
\hline
\hline
Existing & $O\left(n^2\right)$ & $O\left(\min\left(n^2, m\cdot \log n\right)\right)$ & $O(m)$ & \cite{Zhang13}\\
\hline
Our Result & $O(m+n\cdot \alpha(n))$ & $O(m)$ & $O(m)$ & Theorem~\ref{them:bidirected_upper} , Corollary~\ref{cor:bidirected_expected_linear} \\
\hline
\end{tabular}
\end{table}

\begin{table}
\centering
\small
\renewcommand{\arraystretch}{1.5}
\setlength\tabcolsep{2.5pt}
\caption{
Library/Client CFL reachability on the library graph of size $n_1$ and the client graph of size $n_2$.\\
{
$s$ is the number of library summary nodes, as defined in~\cite{Tang15}.\\
$k_1$ is the number of call sites in the library code, with  $k_1<s$.\\
$k_2$ is the number of call sites in the client code.
}
}
\label{tab:library_comparison}
\begin{tabular}{|c||c|c|c|c|c|}
\hline
\textbf{Approach} & \multicolumn{2}{c|}{\textbf{Time}} & \multicolumn{2}{c|}{\textbf{Space}} & Reference\\
\hline
& Library & Client & Library & Client & \\
\hline
\hline
CFL & $O\left(n_1^3\right)$ & $O\left((n_1+n_2)^3\right)$ & $O\left(n_1^2\right)$ & $O\left((n_1+n_2)^2\right )$ & \cite{Tang15} \\
\hline
TAL & $O\left(n_1^6\right)$ & $O\left((s+n_2)^6\right)$ & $O\left(n_1^4\right)$ & $O\left((s+n_2)^4\right )$ & \cite{Tang15} \\
\hline
Our Result & $O\left(n_1 + k_1\cdot \log n_1\right)$ & $O\left(n_2+k_1\cdot \log n_1 + k_2\cdot \log n_2\right)$ & $O(n_1)$ & $O(n_1+n_2)$ & Theorem~\ref{them:library}\\
\hline
\end{tabular}
\end{table}

\noindent{\em Experimental results.}
A key feature of our algorithms are that they are simple to implement. 
We present experimental results both on alias analysis (see Section~\ref{subsec:experiments_bidirected})
and library/client data-dependence analysis (see Section~\ref{subsec:experiments_library})
and show that our algorithms outperform previous approaches for the problems 
on real-world benchmarks.

Due to lack of space, full proofs can be found in full version of this paper~\cite{istreport}.

\subsection{Other Related Work}
We have already discussed in details the most relevant works 
related to language reachability, alias analysis and data-dependence analysis.
We briefly discuss works related to treewidth in program analysis and 
verification.

\noindent{\bf Treewidth in algorithms and program analysis.}
In the context of programming languages, it was shown by~\cite{Thorup98}
that the control-flow graphs for goto-free programs for many programming 
languages have constant treewidth, which has been followed by practical 
approaches as well (such as~\cite{Gustedt02}).
The treewidth property has received a lot of attention in algorithm community,
for NP-complete problems~\cite{Arnborg89,Bern1987216,Bodlaender88}, 
combinatorial optimization problems~\cite{Bertele72}, 
graph problems such as shortest path~\cite{Chaudhuri95,CIP16}.
In algorithmic analysis of programming languages and verification the treewidth
property has been exploited in interprocedural analysis~\cite{CIP15},
concurrent intraprocedural analysis~\cite{CGIP16}, quantitative verification of 
finite-state graphs~\cite{CIP15b}, etc. 
To the best of our knowledge the constant-treewidth property has not be considered
for data-dependence analysis. 
Our experimental results show that in practice many real-world benchmarks 
have the constant-treewidth property, and our algorithms for data-dependence analysis 
exploit this property to present faster algorithms.

\section{Preliminaries}\label{sec:preliminaries}

\noindent{\bf Graphs and paths.}
We denote by $G=(V,E)$ a finite directed graph (henceforth called simply a graph) where $V$ is a set of $n$ nodes and
$E\subseteq V\times V$ is an edge relation of $m$ edges.
Given a set of nodes $X\subseteq V$, we denote by $G\restr{X}=(X, E\cap (X\times X))$
the subgraph of $G$ induced by $X$.
A \emph{path} $P$ is a sequence of edges $(e_1,\dots, e_r)$ and each $e_i=(x_i,y_i)$ is such that $x_1=u$, $y_r=v$, and for all $1 \leq i \leq r-1$ we have $y_i=x_{i+1}$. 
The length of $P$ is $|P|=r$.
A path $P$ is \emph{simple} if no node repeats in $P$ 
(i.e., the path does not contain a cycle). 
Given two paths $P_1=(e_1, \dots e_{r_1})$ and $P_2=(e'_1, \dots e'_{\ell})$ with $e_r=(x,y)$ and $e'_1=(y,z)$,
we denote by $P_1\circ P_2$ the \emph{concatenation} of $P_2$ on $P_1$.
We use the notation $x\in P$ to denote that a node $x$ appears in $P$,
and $e\in P$ to denote that an edge $e$ appears in $P$.
Given a set $\Bag\subseteq V$, we denote by $P\cap \Bag$ the set of nodes of $\Bag$ that appear in $P$.
We say that a node $u$ is \emph{reachable} from node $v$ if there exists a path $P:u\Path v$.

\noindent{\bf Dyck Languages.}
Given a nonnegative integer $k\in \Nats$, we denote by $\Alphabet_k=\{\epsilon\}\cup \{\OpenParenthesis_i, \CloseParenthesis_i\}_{i=1}^k$
a finite \emph{alphabet} of $k$ parenthesis types, together with a null element $\epsilon$.
We denote by $\Dyck_k$ the Dyck language over $\Alphabet_k$,
defined as the language of strings generated by the following context-free grammar $\Grammar_k$:
\[
\StartNonTerminal \to \StartNonTerminal ~ \StartNonTerminal ~ | ~ \OpenNonTerminal_1 ~ \CloseNonTerminal_1 ~ | ~ \dots ~ | ~ \OpenNonTerminal_k ~ \CloseNonTerminal_k ~ | ~\epsilon \ ;
\qquad 
\OpenNonTerminal_i \to \OpenParenthesis_i ~ \StartNonTerminal \ ;\qquad
\CloseNonTerminal_i \to \StartNonTerminal ~ \CloseParenthesis_i
\]
Given a string $s$ and a non-terminal symbol $X$ of the above grammar, we write $X\Produces s$
to denote that $X$ produces $s$ according to the rules of the grammar.
%We consider that $k=O(1)$, i.e., $k$ is fixed wrt to the input size~\cite{Chaudhuri08}.
In the rest of the document we consider an alphabet $\Alphabet_k$ and the corresponding Dyck language $\Dyck_k$.
We also let $\SetOpenParenthesis_k= \{\OpenParenthesis_i\}_{i=1}^k$ and  $\SetClosedParenthesis_k= \{\CloseParenthesis_i\}_{i=1}^k$ 
be the subsets of $\Alphabet_k$ of only opening and closing parenthesis, respectively.

\noindent{\bf Labeled graphs, Dyck reachability, and Dyck SCCs (DSCCs).}
We denote by $G=(V,E)$ a $\Alphabet_k$-labeled directed graph where $V$ is the set of nodes 
and $E\subseteq V\times V \times \Alphabet_k$ is the set of edges labeled with symbols from $\ParSet_k$.
Hence, an edge $e$ is of the form $e=(u,v,\Label)$ where $u,v\in V$ and $\Label\in \ParSet_k$.
We require that for every $u,v\in V$, there is a unique label $\Label$ such that $(u,v,\Label)\in E$.
Often we will be interested only on the endpoints of an edge $e$, in which case we represent $e=(u,v)$,
and will denote by $\Label(e)$ the label of $e$.
Given a path $P$, we define the \emph{label} of $P$ as $\Label(P)=\Label(e_1)\dots \Label(e_r)$.
Given two nodes $u,v$, we say that $v$ is Dyck-reachable from $u$ if there exists a path $P:u\Path v$ such that
$\Label(P)\in \Dyck_k$. In that case, $P$ is called a \emph{witness path} of the reachability.
A set of nodes $X\subseteq V$ is called a \emph{Dyck SCC} (or \emph{DSCC}) if for every pair of nodes
$u,v\in X$, we have that $u$ reaches $v$ and $v$ reaches $u$.
Note that there might exist a DSCC $X$ and a pair of nodes $u,v\in X$ such that
every witness path $P:u\Path v$ might be such that $P\not \subseteq X$, i.e.,
the witness path contains nodes outside the DSCC.
%A DSCC $X$ is \emph{maximal} if for every pair of nodes $u,v\in V$ such that $u$ reaches $v$ and $v$ reaches $u$,
%if $u\in X$ then $v\in X$. In words, a maximal DSCC contains all nodes that are pairwise reachable.

\section{Dyck Reachability on Bidirected Graphs}\label{sec:bidirected}

In this section we present an optimal algorithm for solving the Dyck 
reachability problem on $\Alphabet_k$-labeled \emph{bidirected} graphs $G$.
First, in Section~\ref{subsec:bidirected_problem}, we formally define 
the problem.
Second, in Section~\ref{subsec:bidirected_upper}, we describe an algorithm 
$\BidirectedAlgo$ that solves the problem in time $O(m+n\cdot \alpha(n))$,
where $n$ is the number of nodes of $G$, $m$ is the number of edges of $G$, 
and $\alpha(n)$ is the inverse Ackermann function.
Finally, in Section~\ref{subsec:bidirected_lower}, we present an 
$\Omega(m+n\cdot \alpha(n))$ lower bound.

\subsection{Problem Definition}\label{subsec:bidirected_problem}

We start with the problem definition of Dyck reachability on bidirected graphs.
For the modeling power of bidirected graphs we refer to~\cite{Yuan09,Zhang13}
and our Experimental Section~\ref{subsec:experiments_bidirected}.

\noindent{\bf Bidirected Graphs.}
A $\Alphabet_k$ labeled graph $G=(V,E)$ is called \emph{bidirected} if for 
every pair of nodes $u,v\in V$, the following conditions hold.
(1)~$
(u,v,\epsilon)\in E \quad \text{iff} \quad (v,u,\epsilon)\in E
$; 
and (2)~ for all $1\leq i\leq k$ we have that
$
(u,v,\OpenParenthesis_i)\in E \quad \text{iff} \quad (v,u,\CloseParenthesis_i)\in E
$.
Informally, the edge relation is symmetric, and the labels of symmetric edges 
are complimentary wrt to opening and closing parenthesis.
The following remark captures a key property of bidirected graphs that can be exploited to lead to faster algorithms.

\begin{remark}[\cite{Zhang13}]\label{rem:bidirected_equivalence}
For bidirected graphs the Dyck reachability relation forms an equivalence, i.e., 
for all bidirected graphs $G$, for every pair of nodes $u,v$,
we have that $v$ is Dyck-reachable from $u$ iff $u$ is Dyck-reachable from $v$.
\end{remark}

\begin{remark}\label{rem:bidirected_labels}
We consider without loss of generality that a bidirected graph $G$ 
has no edge $(u,v)$ such that $ \Label(u,v)=\epsilon$, i.e., there are no $\epsilon$-labeled edges.
This is because in such a case, $u,v$ form a DSCC, and can be merged into a single node.
Merging all nodes that share an $\epsilon$-labeled edge requires only linear time,
and hence can be applied as a preprocessing step at (asymptotically) no extra cost.
\end{remark}

\noindent{\bf Dyck reachability on bidirected graphs.}
We are given a $\Alphabet_k$-labeled bidirected graph $G=(V,E)$,
and our task is to compute for every pair of nodes $u,v$ whether $v$ is Dyck-reachable from $u$.
As customary, we consider that $k=O(1)$, i.e., $k$ is fixed wrt to the input graph~\cite{Chaudhuri08}.
In view of Remark~\ref{rem:bidirected_equivalence}, it suffices that the output is a list of DSCCs.
Note that the output has size $\Theta(n)$ instead of $\Theta(n^2)$ that would be required for storing one bit of information per $u,v$ pair. Additionally, the pair query time is $O(1)$, by testing whether the two nodes belong to the same DSCC.

\subsection{An Almost Linear-time Algorithm}\label{subsec:bidirected_upper}
We present our algorithm $\BidirectedAlgo$, for Dyck reachability on bidirected graphs, 
with almost linear-time complexity.

\noindent{\bf Informal description of $\BidirectedAlgo$.}
We start by providing a high-level description of $\BidirectedAlgo$.
The main idea is that for any two distinct nodes $u, v$ to belong to some DSCC $X$,
there must exist two (not necessarily distinct) nodes $x,y$ that belong to some DSCC $Y$ (possibly $X=Y$)\footnote{That is, $x$ and $y$ might refer to the same node, and $X$ and $Y$ to the same DSCC.}
and a closing parenthesis $\CloseParenthesis_i\in \SetClosedParenthesis_k$ such that $(x,u,\CloseParenthesis_i), (y,v,\CloseParenthesis_i)\in E$.
See Figure~\ref{fig:bidirected_principle} for an illustration.
The algorithm uses a Disjoint Sets data structure to maintain DSCCs discovered so far.
Each DSCC is represented as a tree $T$ rooted on some node $x\in V$,
and $x$ is the only node of $T$ that has outgoing edges.
However, any node of $T$ can have incoming edges.
See Figure~\ref{fig:bidirected_state} for an illustration.
Upon discovering that a root node $x$ of some tree $T$ has two or more outgoing edges
$(x,u_1,\CloseParenthesis_i), (x,u_2,\CloseParenthesis_i), \dots (x,u_r,\CloseParenthesis_i) $, for some $\CloseParenthesis_i\in \SetClosedParenthesis_k$,
the algorithm uses $r$ $\Find$ operations of the Disjoint Sets data structure to determine
the trees $T_i$ that the nodes $u_i$ belong to.
Afterwards, a $\Union$ operation is performed between all $T_i$ to form a new tree $T$,
and all the outgoing edges of the root of each $T_i$ are merged to the outgoing edges of the root of $T$.

\begin{figure}[!h]
\begin{subfigure}[b]{.25\textwidth}
\centering
\begin{tikzpicture}[thick, >=latex,
pre/.style={<-,shorten >= 1pt, shorten <=1pt, thick},
post/.style={->,shorten >= 1pt, shorten <=1pt,  thick},
und/.style={very thick, draw=gray},
node1/.style={circle, minimum size=3.5mm, draw=black!100, line width=1pt, inner sep=0},
node2/.style={circle, minimum size=3.5mm, draw=black!100, fill=white!100, very thick, inner sep=0},
virt/.style={circle,draw=black!50,fill=black!20, opacity=0}]

\newcommand{\xdisposition}{0}
\newcommand{\ydisposition}{0}
\newcommand{\xstep}{1}
\newcommand{\ystep}{0.7}
\def\bend{20}

\node	[node2]		(x)	at	(\xdisposition + 0*\xstep, \ydisposition + 0*\ystep)	{$x$};
\node	[node2]		(u)	at	(\xdisposition + -1*\xstep, \ydisposition + 0*\ystep)	{$u$};
\node	[node2]		(v)	at	(\xdisposition + 1*\xstep, \ydisposition + 0*\ystep)	{$v$};
\node	[node2]		(z)	at	(\xdisposition - 2*\xstep, \ydisposition + 0*\ystep)	{$z$};

\draw[->,  very thick, bend right=\bend] (x) to node[above]{$\CloseParenthesis$} (u);
\draw[->,  very thick, bend right=\bend] (u) to node[below]{$\OpenParenthesis$} (x);
\draw[->,  very thick, bend right=\bend] (x) to node[below]{$\CloseParenthesis$} (v);
\draw[->,  very thick, bend right=\bend] (v) to node[above]{$\OpenParenthesis$} (x);
\draw[->,  very thick, bend right=\bend] (u) to node[above]{$\CloseParenthesis$} (z);
\draw[->,  very thick, bend right=\bend] (z) to node[below]{$\OpenParenthesis$} (u);

\draw[->, thick, loop above] (v) to node[above]{$\OpenParenthesis$,$\CloseParenthesis$}   (v);

\end{tikzpicture}
\caption{}
\label{subfig:bidirected_principle1}
\end{subfigure}
\qquad
\begin{subfigure}[b]{.25\textwidth}
\centering
\begin{tikzpicture}[thick, >=latex,
pre/.style={<-,shorten >= 1pt, shorten <=1pt, thick},
post/.style={->,shorten >= 1pt, shorten <=1pt,  thick},
und/.style={very thick, draw=gray},
node1/.style={circle, minimum size=3.5mm, draw=black!100, line width=1pt, inner sep=0},
node2/.style={circle, minimum size=3.5mm, draw=black!100, fill=white!100, very thick, inner sep=0},
virt/.style={circle,draw=black!50,fill=black!20, opacity=0}]

\newcommand{\xdisposition}{0}
\newcommand{\ydisposition}{0}
\newcommand{\xstep}{1}
\newcommand{\ystep}{0.7}
\def\bend{20}

\node	[node2]		(x)	at	(\xdisposition + 0*\xstep, \ydisposition + 0*\ystep)	{$x$};
\node	[node2]		(u)	at	(\xdisposition + -1*\xstep, \ydisposition + 0*\ystep)	{$u$};
\node	[node2]		(v)	at	(\xdisposition + 1*\xstep, \ydisposition + 0*\ystep)	{$v$};
\node	[node2]		(z)	at	(\xdisposition - 2*\xstep, \ydisposition + 0*\ystep)	{$z$};

\draw[->,  very thick, bend right=\bend] (x) to node[above]{$\CloseParenthesis$} (u);
\draw[->,  very thick, bend right=\bend] (u) to node[below]{$\OpenParenthesis$} (x);
\draw[->,  very thick, bend right=\bend] (x) to node[below]{$\CloseParenthesis$} (v);
\draw[->,  very thick, bend right=\bend] (v) to node[above]{$\OpenParenthesis$} (x);
\draw[->,  very thick, bend right=\bend] (u) to node[above]{$\CloseParenthesis$} (z);
\draw[->,  very thick, bend right=\bend] (z) to node[below]{$\OpenParenthesis$} (u);

\draw[->, thick, loop above] (v) to node[above]{$\OpenParenthesis$,$\CloseParenthesis$}   (v);

\draw [black, dashed] plot [smooth cycle] coordinates {(  -1*\xstep-0.2,  + 0*\ystep-0.3) ( -1*\xstep,  + 0*\ystep-0.3+1.3) ( -1*\xstep+2.4,  + 0*\ystep-0.3+1.3) ( -1*\xstep+2.2,  + 0*\ystep-0.3) (-1*\xstep+1.7,  + 0*\ystep-0.3) ( -1*\xstep+1.7,  + 0*\ystep-0.3+0.9) ( -1*\xstep+0.3,  + 0*\ystep-0.3+0.9) ( -1*\xstep+0.3,  + 0*\ystep-0.3) };

\end{tikzpicture}

\caption{}
\label{subfig:bidirected_principle2}
\end{subfigure}
\qquad
\begin{subfigure}[b]{.25\textwidth}
\centering
\begin{tikzpicture}[thick, >=latex,
pre/.style={<-,shorten >= 1pt, shorten <=1pt, thick},
post/.style={->,shorten >= 1pt, shorten <=1pt,  thick},
und/.style={very thick, draw=gray},
node1/.style={circle, minimum size=3.5mm, draw=black!100, line width=1pt, inner sep=0},
node2/.style={circle, minimum size=3.5mm, draw=black!100, fill=white!100, very thick, inner sep=0},
virt/.style={circle,draw=black!50,fill=black!20, opacity=0}]

\newcommand{\xdisposition}{0}
\newcommand{\ydisposition}{0}
\newcommand{\xstep}{1}
\newcommand{\ystep}{0.7}
\def\bend{20}

\node	[node2]		(x)	at	(\xdisposition + 0*\xstep, \ydisposition + 0*\ystep)	{$x$};
\node	[node2]		(u)	at	(\xdisposition + -1*\xstep, \ydisposition + 0*\ystep)	{$u$};
\node	[node2]		(v)	at	(\xdisposition + 1*\xstep, \ydisposition + 0*\ystep)	{$v$};
\node	[node2]		(z)	at	(\xdisposition - 2*\xstep, \ydisposition + 0*\ystep)	{$z$};

\draw[->,  very thick, bend right=\bend] (x) to node[above]{$\CloseParenthesis$} (u);
\draw[->,  very thick, bend right=\bend] (u) to node[below]{$\OpenParenthesis$} (x);
\draw[->,  very thick, bend right=\bend] (x) to node[below]{$\CloseParenthesis$} (v);
\draw[->,  very thick, bend right=\bend] (v) to node[above]{$\OpenParenthesis$} (x);
\draw[->,  very thick, bend right=\bend] (u) to node[above]{$\CloseParenthesis$} (z);
\draw[->,  very thick, bend right=\bend] (z) to node[below]{$\OpenParenthesis$} (u);

\draw[->, thick, loop above] (v) to node[above]{$\OpenParenthesis$,$\CloseParenthesis$}   (v);

\draw [black, dashed] plot [smooth cycle] coordinates {(  -2*\xstep-0.2,  + 0*\ystep-0.4) ( -2*\xstep,  + 0*\ystep-0.3+1.3) ( -1*\xstep+2.4,  + 0*\ystep-0.3+1.3) ( -1*\xstep+2.2,  + 0*\ystep-0.3) (-1*\xstep+1.7,  + 0*\ystep-0.3) ( -1*\xstep+1.7,  + 0*\ystep-0.3+0.9) ( -1*\xstep+0.3,  + 0*\ystep-0.3+0.9) ( -1*\xstep+0.3,  + 0*\ystep-0.4) };

\end{tikzpicture}

\caption{}
\label{subfig:bidirected_principle3}
\end{subfigure}

\caption{Illustration of the merging principle of $\BidirectedAlgo$.
\textbf{(\ref{subfig:bidirected_principle1})} The nodes $u$ and $v$ are in the same DSCC since node $x$ has an outgoing edge to each of $u$ and $v$ labeled with the closing parenthesis $\CloseParenthesis$.
\textbf{(\ref{subfig:bidirected_principle2})} Similarly, nodes $z$ and $v$ belong to the same DSCC, since there exist two nodes $u$ and $v$ such that (i)~ $u$ and $v$ belong to the same DSCC, (ii)~$u$ has an outgoing edge to $z$, and $v$ has an outgoing edge to itself, and (iii)~both outgoing edges are labeled with the same closing parenthesis symbol.
\textbf{(\ref{subfig:bidirected_principle3})} The final DSCC formation.
}
\label{fig:bidirected_principle}
\end{figure}
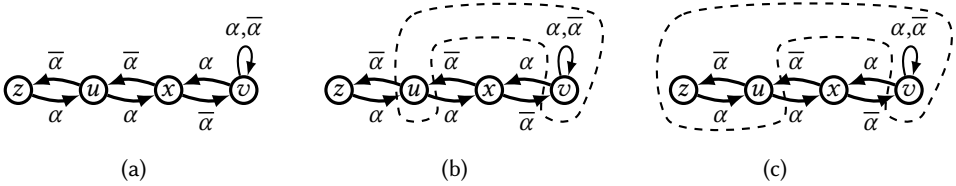

\begin{figure}[!h]
\centering
\begin{tikzpicture}[thick, >=latex,
pre/.style={<-,shorten >= 1pt, shorten <=1pt, thick},
post/.style={->,shorten >= 1pt, shorten <=1pt,  thick},
und/.style={very thick, draw=gray},
node1/.style={circle, minimum size=3.5mm, draw=black!100, line width=1pt, inner sep=0},
node2/.style={circle, minimum size=3.5mm, draw=black!100, fill=white!100, very thick, inner sep=0},
virt/.style={circle,draw=black!50,fill=black!20, opacity=0}]

\newcommand{\xdisposition}{0}
\newcommand{\ydisposition}{0}
\newcommand{\xstep}{0.45}
\newcommand{\ystep}{0.7}

\node	[node2]		(1)	at	(\xdisposition + 0*\xstep, \ydisposition + 0*\ystep)	{$s$};
\node	[node2]		(2)	at	(\xdisposition + -2*\xstep, \ydisposition + -1*\ystep)	{};
\node	[node2]		(3)	at	(\xdisposition + -0*\xstep, \ydisposition + -1*\ystep)	{};
\node	[node2]		(4)	at	(\xdisposition + 2*\xstep, \ydisposition + -1*\ystep)	{};
\node	[node2]		(5)	at	(\xdisposition + -2*\xstep, \ydisposition + -2*\ystep)	{};
\node	[node2]		(6)	at	(\xdisposition + -1.5*\xstep, \ydisposition + -3*\ystep)	{$t$};
\node	[node2]		(7)	at	(\xdisposition + -2.5*\xstep, \ydisposition + -3*\ystep)	{$u$};
\node	[node2]		(8)	at	(\xdisposition + -0.5*\xstep, \ydisposition + -2*\ystep)	{};
\node	[node2]		(9)	at	(\xdisposition + 0.5*\xstep, \ydisposition + -2*\ystep)	{};

\draw[-,  very thick] (1) to (2);
\draw[-,  very thick] (1) to (3);
\draw[-,  very thick] (1) to (4);
\draw[-,  very thick] (2) to (5);
\draw[-,  very thick] (5) to (6);
\draw[-, very  thick] (5) to (7);
\draw[-,  very thick] (3) to (8);
\draw[-,  very thick] (3) to (9);

\renewcommand{\xdisposition}{4}

\node	[node2]		(19)	at	(\xdisposition + 0*\xstep, \ydisposition + 0*\ystep)	{$v$};
\node	[node2]		(20)	at	(\xdisposition + 0*\xstep, \ydisposition + -1*\ystep)	{};
\node	[node2]		(22)	at	(\xdisposition + -1.5*\xstep, \ydisposition + -1*\ystep)	{$w$};
\node	[node2]		(23)	at	(\xdisposition + 1.5*\xstep, \ydisposition + -1*\ystep)	{$x$};
%\node	[node2]		(24)	at	(\xdisposition + -3*\xstep, \ydisposition + -1*\ystep)	{24};
\node	[node2]		(25)	at	(\xdisposition + -2*\xstep, \ydisposition + -2*\ystep)	{};
\node	[node2]		(26)	at	(\xdisposition + -1*\xstep, \ydisposition + -2*\ystep)	{};
\node	[node2]		(27)	at	(\xdisposition + 1.5*\xstep, \ydisposition + -2*\ystep)	{$y$};

\draw[-,  very thick] (19) to (20);
\draw[-,  very thick] (19) to (22);
\draw[-,  very thick] (19) to (23);
%\draw[-,  very thick] (19) to (24);
\draw[-,  very thick] (22) to (25);
\draw[-,  very thick] (22) to (26);
\draw[-,  very thick] (23) to (27);

\renewcommand{\xdisposition}{7.5}

\node	[node2]		(28)	at	(\xdisposition + 0*\xstep, \ydisposition + 0*\ystep)	{$z$};
\node	[node2]		(29)	at	(\xdisposition + 0*\xstep, \ydisposition + -1*\ystep)	{};
\node	[node2]		(30)	at	(\xdisposition + -0.5*\xstep, \ydisposition + -2*\ystep)	{};
\node	[node2]		(31)	at	(\xdisposition + 1.5*\xstep, \ydisposition + -1*\ystep)	{};
\node	[node2]		(32)	at	(\xdisposition + 0.5*\xstep, \ydisposition + -2*\ystep)	{};
\node	[node2]		(33)	at	(\xdisposition - 1.5*\xstep, \ydisposition + -1*\ystep)	{};
\node	[node2]		(34)	at	(\xdisposition - 1.5*\xstep, \ydisposition + -2*\ystep)	{};

\draw[-,  very thick] (28) to (29);
\draw[-,  very thick] (28) to (31);
\draw[-,  very thick] (29) to (30);
\draw[-,  very thick] (29) to (32);
\draw[-,  very thick] (28) to (33);
\draw[-,  very thick] (33) to (34);

\draw[->, thick, out=180, in=90] (1) -- (-2.2*\xstep, + -0.3*\ystep) to node[left]{$\CloseParenthesis_1$}   (7);
\draw[->, thick, out=0, in=90] (1) to node[above]{$\CloseParenthesis_2$} (22);
\draw[->, thick, out=0, in=150] (1) to node[above]{$\CloseParenthesis_1$}  (28);

\draw[->, thick, out=160, in=340] (19) to  node[below, pos=0.5, label={[yshift=-20]$\CloseParenthesis_2$}]{} (6);

\draw[->, thick, out=180, in=50] (28) to node[above, pos=0.6] {$\CloseParenthesis_3$}  (23);
\draw[->, thick, out=180, in=50] (28) to node[below, pos=0.6] {$\CloseParenthesis_3$} (27);
\draw[->, thick, loop above] (28) to node[above]{$\CloseParenthesis_2$}   (28);

\end{tikzpicture}
\caption{A state of $\BidirectedAlgo$ consists of a set of trees, with outgoing edges coming only from the root of each tree.}\label{fig:bidirected_state}
\end{figure}
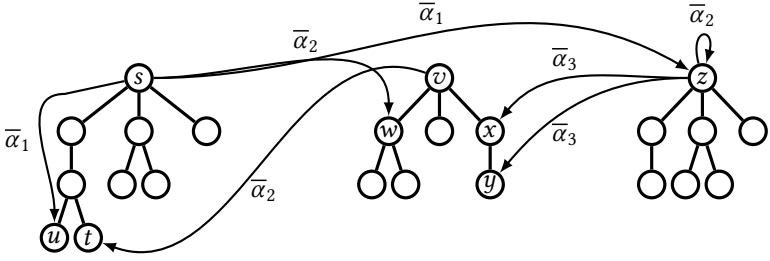

\noindent{\bf Complexity overview.}
The cost of every $\Find$ and $\Union$ operation is bounded by the inverse Ackermann function $\alpha(n)$ (see~\cite{Tarjan75}), which, for all practical purposes, can be considered constant.
Additionally, every edge-merge operation requires constant time, using a linked list for storing the outgoing edges.
Although list merging in constant time creates the possibility of duplicate edges,
such duplicates come at no additional complexity cost.
Since every $\Union$ of $k$ trees reduces the number of existing edges by $k-1$, 
the overall complexity of $\BidirectedAlgo$ is $O(m\cdot \alpha(n))$.
We later show how to obtain the $O(m + n \cdot \alpha(n))$ complexity.

We are now ready to give the formal description of $\BidirectedAlgo$.
We start with introducing the Union-Find problem, and its solution given by a disjoint sets data structure.

\noindent{\bf The Union-Find problem.}
The Union-Find problem is a well-studied problem in the area of algorithms and data structures~\cite{Galil91,Cormen01}.
The problem is defined over a \emph{universe} $X$ of  $n$ elements, and the task is to maintain partitions of $X$ under set union operations.
Initially, every element $x\in X$ belongs to a singleton set $\{x\}$.
A \emph{union-find} sequence $\sigma$ is a sequence of $m$ (typically $m\geq n$) operations of the following two types.
\begin{compactenum}
\item $\Union(x,y)$, for $x,y\in X$, performs a union of the sets that $x$ and $y$ belong to.
\item $\Find(x)$, for $x\in X$, returns the name of the unique set containing $x$.
\end{compactenum}
The sequence $\sigma$ is presented online, i.e., an operation needs to be completed before the next one is revealed.
Additionally, a $\Union(x,y)$ operation is allowed in the $i$-th position of $\sigma$ only if the prefix of $\sigma$ up to position $i-1$
places $x$ and $y$ on different sets.
The output of the problem consists of the answers to $\Find$ operations of $\sigma$.
It is known that the problem can be solved in $O(m\cdot \alpha(n))$ time, by an appropriate Disjoint Sets data structure~\cite{Tarjan75},
and that this complexity is optimal~\cite{Tarjan79,Banachowski80}.
%in various models of computation~\cite{Tarjan79,Banachowski80,Fredman89,LaPoutre96}.

\noindent{\bf The $\DisjointSets$ data structure.}
We consider at our disposal a Disjoint Sets data structure $\DisjointSets$ which maintains a set of subsets of $V$
under a sequence of set union operations.
At all times, the name of each set $X$ is a node $x\in X$ which is considered to be the representative of $x$.
$\DisjointSets$ provides the following operations.
\begin{compactenum}
\item For a node $u$, $\MakeSet(u)$ constructs the singleton set $\{u\}$.
\item For a node $u$, $\Find(u)$ returns the representative of the set that $u$ belongs to.
\item For a set of nodes $S\subseteq V$ which are pairwise in different sets, and a distinguished node $x\in S$,
 $\Union(S,x)$ performs the union of the sets that the nodes in $S$ belong to, and makes $x$ the representative of the new set.
\end{compactenum}
The $\DisjointSets$ data structure can be straightforwardly obtained from the corresponding Disjoint Sets data structures used to solve the Union-Find problem~\cite{Tarjan75}, and has $O(\alpha(n))$ amortized complexity per operation.
Typically each set is stored as a rooted tree, and the root node is the representative of the set.

\noindent{\bf Formal description of $\BidirectedAlgo$.}
We are now ready to present formally $\BidirectedAlgo$ in Algorithm~\ref{algo:bidirectedalgo}.
Recall that, in view of Remark~\ref{rem:bidirected_labels}, we consider that the input graph has no $\epsilon$-labeled edges.
In the initialization phase, the algorithm constructs a map $\Edges:~V\times \SetClosedParenthesis_k\to V^*$.
For each node $u\in V$ and closing parenthesis $\CloseParenthesis_i\in \SetClosedParenthesis_k$, 
$\Edges[u][\CloseParenthesis_i]$ will store the nodes that are found to be reachable from $u$
via a path $P$ such that $\CloseNonTerminal_i \Produces \Label(P)$
(i.e., the label of $P$ has matching parenthesis except for the last parenthesis $\CloseParenthesis_i$).
Observe that all such nodes must belong to the same DSCC.

The main computation happens in the loop of Line~\ref{line:loop_outer}.
The algorithm maintains a queue $\Queue$ that acts as a worklist and stores pairs $(u, \CloseParenthesis_i)$
such that $u$ is a node that has been found to contain at least two outgoing edges labeled with $\CloseParenthesis_i$.
Upon extracting an element $(u, \CloseParenthesis_i)$ from the queue,
the algorithm obtains the representatives $v$ of the sets of the nodes in $\Edges[u][\CloseParenthesis_i]$.
Since all such nodes belong to the same DSCC, the algorithm chooses an element $x$ to be the new representative,
and performs a $\Union$ operation of the underlying sets.
The new representative $x$ gathers the outgoing edges of all other nodes $v\in \Edges[u][\CloseParenthesis_i]$,
and afterwards $\Edges[u][\CloseParenthesis_i]$ points only to $x$.
%We note that our requirement for specifying the representative of a set union operation (i.e., the $x$ in $\Union(S,x)$)
%is only so that $x$ and $u$ are distinct nodes in the the loop of Line~\ref{line:loop_outer},
%which makes the algorithm easier to present.

\begin{algorithm}%[H]
\small
\DontPrintSemicolon
\SetKwFunction{shorten}{$\Shorten$}
\caption{$\BidirectedAlgo$}\label{algo:bidirectedalgo}
\KwIn{A $\Alphabet_k$-labeled bidirected graph $G=(V,E)$}
\KwOut{A $\DisjointSets$ map of DSCCs}
\BlankLine
\tcp{Initialization}
$\Queue\gets $ an empty queue\label{line:init_begin}\\
$\Edges\gets$ a map $V\times \SetClosedParenthesis_k\to V^*$ implemented as a linked list\\
$\DisjointSets\gets$ a disjoint-sets data structure over $V$\\
\ForEach{$u\in V$}{
$\DisjointSets.\MakeSet(u)$\\
\For{$i \gets 1$ \textbf{to} $k$}{
$\Edges[u][\CloseParenthesis_i]\gets (v:~(u,v,\CloseParenthesis_i)\in E)$\\
\lIf{$|\Edges[u][\CloseParenthesis_i]|\geq 2$}{
Insert $(u, \CloseParenthesis_i)$ in $\Queue$\label{line:queue_insert1}
}
}
}\label{line:init_end}
\tcp{Computation}
\While{$\Queue$ is not empty}{\label{line:loop_outer}
Extract $(u, \CloseParenthesis_i)$ from $\Queue$\label{line:extract}\\
\uIf{$u=\DisjointSets.\Find(u)$}{\label{line:if_parent}
Let $S\gets \{\DisjointSets.\Find(w):w\in \Edges[u][\CloseParenthesis_i]\}$\label{line:construct_s}\\
\eIf{$|S|\geq 2$}{\label{line:if_two}
Let $x\gets $ some arbitrary element of $S\setminus\{u\}$\label{line:choose_x}\\
Make $\DisjointSets.\Union(S,x)$\label{line:update_scc}\\
\For{$j \gets 1$ \textbf{to} $k$}{\label{line:loop_append}
\ForEach{$v\in S\setminus\{x\}$}{\label{line:loop_inner}
\eIf{$u\neq v$ or $i\neq j$}{\label{line:if_move}
Move $\Edges[v][\CloseParenthesis_j]$ to $\Edges[x][\CloseParenthesis_j]$\label{line:append}\\
}
{
Append $(x)$ to $\Edges[x][\CloseParenthesis_j]$\label{line:append2}
}
}
\lIf{$|\Edges[x][\CloseParenthesis_j]|\geq 2$}{
Insert $(x, \CloseParenthesis_j)$ in $\Queue$\label{line:queue_insert2}
}
}
}
{
Let $x\gets$ the single node in $S$
}
\lIf{$u\not \in S$ or $|S|=1$}{\label{line:if_update}
$\Edges[u][\CloseParenthesis_i]\gets (x)$\label{line:new_edges}
}
}
}
\Return{$\DisjointSets$}
\end{algorithm}

\noindent{\bf Example.}
Consider the state of the algorithm given by Figure~\ref{fig:bidirected_state}, representing the DSCCs of the Union-Find data structure $\DisjointSets$ (i.e., the undirected trees in the figure) as well as the contents of the $\Edges$ data structure (i.e., the directed edges in the Figure).
There are currently 3 DSCCS, with representatives $s$, $v$ and $z$.
Recall that the queue $\Queue$ stores (node, closing parenthesis) pairs with the property that the node has at least two outgoing edges labeled with the respective closing parenthesis.
Observe that nodes $s$ and $z$ have at two outgoing edges each that have the same type of parenthesis,
hence they must have been inserted in the queue $\Queue$ at some point.
Assume that $\Queue=[(s, \CloseParenthesis_1), (z,\CloseParenthesis_3)]$.
The algorithm will exhibit the following sequence of steps.

\begin{compactenum}
\item The element $(z,\CloseParenthesis_3)$ is extracted from $\Queue$.
We have $\Edges[z][\CloseParenthesis_3]=(x,y)$.
Observe that $x$ and $y$ belong to the same DSCC rooted at $v$, hence in Line~\ref{line:construct_s} the algorithm will construct
$S=\{v\}$.
Since $|S|=1$, the algorithm will simply set $\Edges[z][\CloseParenthesis_3]=(v)$ in Line~\ref{line:if_update}, and no new DSCC has been formed.

\item The element $(s, \CloseParenthesis_1)$ is extracted from $\Queue$.
We have $\Edges[s][\CloseParenthesis_1]=(u,z)$.
Since $u$ and $z$ belong to different DSCCs, the algorithm will construct
$S=\{s, z\}$, and perform a $\DisjointSets.\Union(S,x)$ operation, where $x=z$.
Note that union-by-rank will make the tree of $z$ a subtree of the tree of $s$, i.e., $z$ will become a child of $s$.
Afterwards, the algorithm swaps the names of $z$ and $s$, as required by the choice of $x$ in Line~\ref{line:choose_x}.
Finally, in Line~\ref{line:append}, the algorithm will move $\Edges[s][\CloseParenthesis_i]$ to $\Edges[z][\CloseParenthesis_i]$ for $i=1,2$.
Since now $|\Edges[z][\CloseParenthesis_2]|\geq 2$, the algorithm inserts $(z, \CloseParenthesis_2)$ in $\Queue$.
See Figure~\ref{subfig:example1}.

\item The element $(z, \CloseParenthesis_2)$ is extracted from $\Queue$.
We have $\Edges[z][\CloseParenthesis_2]=(v,z)$.
Since $v$ and $z$ belong to different DSCCs, the algorithm will construct
$S=\{v, z\}$, and perform a $\DisjointSets.\Union(S,x)$ operation, where $x=v$.
Note that union-by-rank will make the tree of $v$ a subtree of the tree of $z$, i.e., $v$ will become a child of $z$.
Afterwards, the algorithm swaps the names of $v$ and $z$, as required by the choice of $x$ in Line~\ref{line:choose_x}.
Finally, in Line~\ref{line:append}, the algorithm will move $\Edges[z][\CloseParenthesis_2]$ to $\Edges[v][\CloseParenthesis_2]$.
Since now $|\Edges[v][\CloseParenthesis_2]|\geq 2$, the algorithm inserts $(v, \CloseParenthesis_2)$ in $\Queue$.
See Figure~\ref{subfig:example2}.

\item The element $(v, \CloseParenthesis_2)$ is extracted from $\Queue$.
We have $\Edges[v][\CloseParenthesis_2]=(v,t)$.
Observe that $v$ and $t$ belong to the same DSCC rooted at $v$, hence in Line~\ref{line:construct_s} the algorithm will construct
$S=\{v\}$.
Since $|S|=1$, the algorithm will simply set $\Edges[v][\CloseParenthesis_2]=(v)$ in Line~\ref{line:if_update}, and will terminate.
\end{compactenum}

\begin{figure}[!h]
\begin{subfigure}[b]{.48\textwidth}
\centering
\begin{tikzpicture}[thick, >=latex,
pre/.style={<-,shorten >= 1pt, shorten <=1pt, thick},
post/.style={->,shorten >= 1pt, shorten <=1pt,  thick},
und/.style={very thick, draw=gray},
node1/.style={circle, minimum size=3.5mm, draw=black!100, line width=1pt, inner sep=0},
node2/.style={circle, minimum size=3.5mm, draw=black!100, fill=white!100, very thick, inner sep=0},
virt/.style={circle,draw=black!50,fill=black!20, opacity=0}]

\newcommand{\xdisposition}{0}
\newcommand{\ydisposition}{0}
\newcommand{\xstep}{0.45}
\newcommand{\ystep}{0.7}

\node	[node2]		(1)	at	(\xdisposition + 0*\xstep, \ydisposition + 0*\ystep)	{$z$};
\node	[node2]		(2)	at	(\xdisposition + -2*\xstep, \ydisposition + -1*\ystep)	{};
\node	[node2]		(3)	at	(\xdisposition + -0*\xstep, \ydisposition + -1*\ystep)	{};
\node	[node2]		(4)	at	(\xdisposition + 2*\xstep, \ydisposition + -1*\ystep)	{};
\node	[node2]		(5)	at	(\xdisposition + -2*\xstep, \ydisposition + -2*\ystep)	{};
\node	[node2]		(6)	at	(\xdisposition + -1.5*\xstep, \ydisposition + -3*\ystep)	{$t$};
\node	[node2]		(7)	at	(\xdisposition + -2.5*\xstep, \ydisposition + -3*\ystep)	{$u$};
\node	[node2]		(8)	at	(\xdisposition + -0.5*\xstep, \ydisposition + -2*\ystep)	{};
\node	[node2]		(9)	at	(\xdisposition + 0.5*\xstep, \ydisposition + -2*\ystep)	{};

\draw[-,  very thick] (1) to (2);
\draw[-,  very thick] (1) to (3);
\draw[-,  very thick] (1) to (4);
\draw[-,  very thick] (2) to (5);
\draw[-,  very thick] (5) to (6);
\draw[-, very  thick] (5) to (7);
\draw[-,  very thick] (3) to (8);
\draw[-,  very thick] (3) to (9);

\renewcommand{\xdisposition}{3}

\node	[node2]		(19)	at	(\xdisposition + 0*\xstep, \ydisposition + 0*\ystep)	{$v$};
\node	[node2]		(20)	at	(\xdisposition + 0*\xstep, \ydisposition + -1*\ystep)	{};
\node	[node2]		(22)	at	(\xdisposition + -1.5*\xstep, \ydisposition + -1*\ystep)	{$w$};
\node	[node2]		(23)	at	(\xdisposition + 1.5*\xstep, \ydisposition + -1*\ystep)	{$x$};
%\node	[node2]		(24)	at	(\xdisposition + -3*\xstep, \ydisposition + -1*\ystep)	{24};
\node	[node2]		(25)	at	(\xdisposition + -2*\xstep, \ydisposition + -2*\ystep)	{};
\node	[node2]		(26)	at	(\xdisposition + -1*\xstep, \ydisposition + -2*\ystep)	{};
\node	[node2]		(27)	at	(\xdisposition + 1.5*\xstep, \ydisposition + -2*\ystep)	{$y$};

\draw[-,  very thick] (19) to (20);
\draw[-,  very thick] (19) to (22);
\draw[-,  very thick] (19) to (23);
%\draw[-,  very thick] (19) to (24);
\draw[-,  very thick] (22) to (25);
\draw[-,  very thick] (22) to (26);
\draw[-,  very thick] (23) to (27);

\renewcommand{\xdisposition}{-2.2}
\renewcommand{\ydisposition}{-1*\ystep}

\node	[node2]		(28)	at	(\xdisposition + 0*\xstep, \ydisposition + 0*\ystep)	{$s$};
\node	[node2]		(29)	at	(\xdisposition + 0*\xstep, \ydisposition + -1*\ystep)	{};
\node	[node2]		(30)	at	(\xdisposition + -0.5*\xstep, \ydisposition + -2*\ystep)	{};
\node	[node2]		(31)	at	(\xdisposition + 1.5*\xstep, \ydisposition + -1*\ystep)	{};
\node	[node2]		(32)	at	(\xdisposition + 0.5*\xstep, \ydisposition + -2*\ystep)	{};
\node	[node2]		(33)	at	(\xdisposition - 1.5*\xstep, \ydisposition + -1*\ystep)	{};
\node	[node2]		(34)	at	(\xdisposition - 1.5*\xstep, \ydisposition + -2*\ystep)	{};

\draw[-,  very thick] (1) to (28);
\draw[-,  very thick] (28) to (29);
\draw[-,  very thick] (28) to (31);
\draw[-,  very thick] (29) to (30);
\draw[-,  very thick] (29) to (32);
\draw[-,  very thick] (28) to (33);
\draw[-,  very thick] (33) to (34);

\draw[->, thick, loop above] (1) to node[above]{$\CloseParenthesis_1$,$\CloseParenthesis_2$}   (1);
\draw[->, thick, out=0, in=90] (1) to node[above]{$\CloseParenthesis_2$} (22);
%\draw[->, thick, out=0, in=150] (1) to node[above]{$\CloseParenthesis_1$}  (28);

\draw[->, thick, out=160, in=340] (19) to  node[below, pos=0.5, label={[yshift=-20]$\CloseParenthesis_2$}]{} (6);

\draw[->, thick, out=20, in=100] (1) to node[above, pos=0.6] {$\CloseParenthesis_3$}  (19);

\end{tikzpicture}
\caption{}
\label{subfig:example1}
\end{subfigure}
\hspace{3.5mm}
\begin{subfigure}[b]{.48\textwidth}
\centering
\begin{tikzpicture}[thick, >=latex,
pre/.style={<-,shorten >= 1pt, shorten <=1pt, thick},
post/.style={->,shorten >= 1pt, shorten <=1pt,  thick},
und/.style={very thick, draw=gray},
node1/.style={circle, minimum size=3.5mm, draw=black!100, line width=1pt, inner sep=0},
node2/.style={circle, minimum size=3.5mm, draw=black!100, fill=white!100, very thick, inner sep=0},
virt/.style={circle,draw=black!50,fill=black!20, opacity=0}]

\newcommand{\xdisposition}{0}
\newcommand{\ydisposition}{0}
\newcommand{\xstep}{0.45}
\newcommand{\ystep}{0.7}

\node	[node2]		(1)	at	(\xdisposition + 0*\xstep, \ydisposition + 0*\ystep)	{$v$};
\node	[node2]		(2)	at	(\xdisposition + -2*\xstep, \ydisposition + -1*\ystep)	{};
\node	[node2]		(3)	at	(\xdisposition + -0*\xstep, \ydisposition + -1*\ystep)	{};
\node	[node2]		(4)	at	(\xdisposition + 2*\xstep, \ydisposition + -1*\ystep)	{};
\node	[node2]		(5)	at	(\xdisposition + -2*\xstep, \ydisposition + -2*\ystep)	{};
\node	[node2]		(6)	at	(\xdisposition + -1.5*\xstep, \ydisposition + -3*\ystep)	{$t$};
\node	[node2]		(7)	at	(\xdisposition + -2.5*\xstep, \ydisposition + -3*\ystep)	{$u$};
\node	[node2]		(8)	at	(\xdisposition + -0.5*\xstep, \ydisposition + -2*\ystep)	{};
\node	[node2]		(9)	at	(\xdisposition + 0.5*\xstep, \ydisposition + -2*\ystep)	{};

\draw[-,  very thick] (1) to (2);
\draw[-,  very thick] (1) to (3);
\draw[-,  very thick] (1) to (4);
\draw[-,  very thick] (2) to (5);
\draw[-,  very thick] (5) to (6);
\draw[-, very  thick] (5) to (7);
\draw[-,  very thick] (3) to (8);
\draw[-,  very thick] (3) to (9);

\renewcommand{\xdisposition}{2.2}
\renewcommand{\ydisposition}{-1*\ystep}

\node	[node2]		(19)	at	(\xdisposition + 0*\xstep, \ydisposition + 0*\ystep)	{$z$};
\node	[node2]		(20)	at	(\xdisposition + 0*\xstep, \ydisposition + -1*\ystep)	{};
\node	[node2]		(22)	at	(\xdisposition + -1.5*\xstep, \ydisposition + -1*\ystep)	{$w$};
\node	[node2]		(23)	at	(\xdisposition + 1.5*\xstep, \ydisposition + -1*\ystep)	{$x$};
%\node	[node2]		(24)	at	(\xdisposition + -3*\xstep, \ydisposition + -1*\ystep)	{24};
\node	[node2]		(25)	at	(\xdisposition + -2*\xstep, \ydisposition + -2*\ystep)	{};
\node	[node2]		(26)	at	(\xdisposition + -1*\xstep, \ydisposition + -2*\ystep)	{};
\node	[node2]		(27)	at	(\xdisposition + 1.5*\xstep, \ydisposition + -2*\ystep)	{$y$};

\draw[-,  very thick] (19) to (20);
\draw[-,  very thick] (19) to (22);
\draw[-,  very thick] (19) to (23);
%\draw[-,  very thick] (19) to (24);
\draw[-,  very thick] (22) to (25);
\draw[-,  very thick] (22) to (26);
\draw[-,  very thick] (23) to (27);

\renewcommand{\xdisposition}{-2.2}
\renewcommand{\ydisposition}{-1*\ystep}

\node	[node2]		(28)	at	(\xdisposition + 0*\xstep, \ydisposition + 0*\ystep)	{$s$};
\node	[node2]		(29)	at	(\xdisposition + 0*\xstep, \ydisposition + -1*\ystep)	{};
\node	[node2]		(30)	at	(\xdisposition + -0.5*\xstep, \ydisposition + -2*\ystep)	{};
\node	[node2]		(31)	at	(\xdisposition + 1.5*\xstep, \ydisposition + -1*\ystep)	{};
\node	[node2]		(32)	at	(\xdisposition + 0.5*\xstep, \ydisposition + -2*\ystep)	{};
\node	[node2]		(33)	at	(\xdisposition - 1.5*\xstep, \ydisposition + -1*\ystep)	{};
\node	[node2]		(34)	at	(\xdisposition - 1.5*\xstep, \ydisposition + -2*\ystep)	{};

\draw[-,  very thick] (1) to (19);
\draw[-,  very thick] (1) to (28);
\draw[-,  very thick] (28) to (29);
\draw[-,  very thick] (28) to (31);
\draw[-,  very thick] (29) to (30);
\draw[-,  very thick] (29) to (32);
\draw[-,  very thick] (28) to (33);
\draw[-,  very thick] (33) to (34);

\draw[->, thick, loop above] (1) to node[above]{$\CloseParenthesis_1$,$\CloseParenthesis_2$, $\CloseParenthesis_3$}   (1);

\node[minimum size=0, inner sep=0] (t) at (1.6*\xstep, + -2.2*\ystep) {};
\draw[-, thick, thick,  out=-55, in=80] (1) to (t);
\draw[->, thick,  out=-115, in=20] (t)  to node[below]{$\CloseParenthesis_2$}   (6);
%\draw[->, thick, out=-45, in=0] (1) to node[below]{$\CloseParenthesis_2$}   (6);

\node[] at (0,-2.6) {};

\end{tikzpicture}
\caption{}
\label{subfig:example2}
\end{subfigure}
\caption{The intermediate stages of $\BidirectedAlgo$ starting from the stage of Figure~\ref{fig:bidirected_state}.}
\label{fig:bidirected_example}
\end{figure}
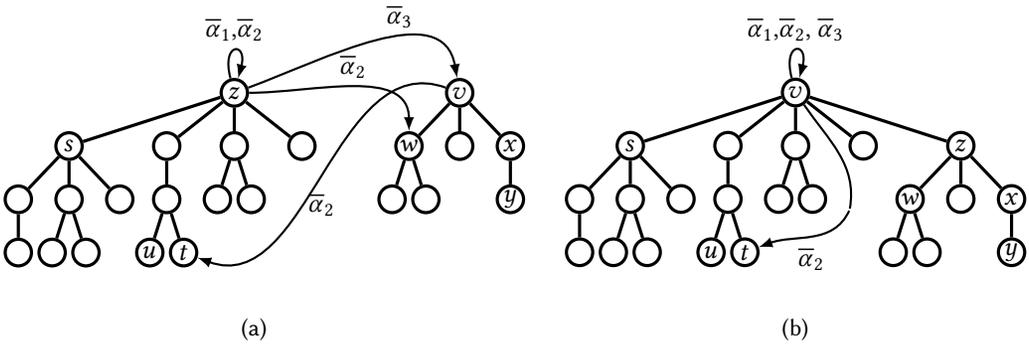

\noindent{\bf Correctness.}
We start with the correctness statement of $\BidirectedAlgo$,
which is established in two parts, namely the soundness and completeness,
which are shown in the following two lemmas.

\begin{lemma}[Soundness]\label{lem:bidirected_soundness}
At the end of $\BidirectedAlgo$, for every pair of nodes $u,v\in V$,
if $\DisjointSets.\Find(u)=\DisjointSets.Find(v)$ then $u$ and $v$ belong to the same DSCC.
\end{lemma}

\begin{lemma}[Completeness]\label{lem:bidirected_completeness}
At the end of $\BidirectedAlgo$, for every pair of nodes $u,v\in V$ in the same DSCC,
$u$ and $v$ belong to the same set of $\DisjointSets$.
\end{lemma}

\noindent{\bf Complexity.}
We now establish the complexity of $\BidirectedAlgo$, in a sequence of lemmas.

\begin{lemma}\label{lem:outer_loop_bound}
The main loop of Line~\ref{line:loop_outer} will be executed $O(n)$ times.
\end{lemma}
\begin{proof}
Initially $\Queue$ is populated by Line~\ref{line:queue_insert1}, which inserts $O(n)$ elements, as $k=O(1)$.
Afterwards, for every $\ell\leq k=O(1)$ elements $(u, \CloseParenthesis_j)$ inserted in $\Queue$ via Line~\ref{line:queue_insert2},
there is at least one node $v\in S$ which stops being a representative of its own set in $\DisjointSets$, and thus will not be in $S$ in further iterations.
Hence $\Queue$ will contain $O(n)$ elements in total, and the result follows.
\end{proof}

\noindent{\bf The sets $S_j$ and $S'_j$.}
Consider an element $(u, \CloseParenthesis_i)$ extracted from $\Queue$ in the $j$-th iteration of the algorithm in Line~\ref{line:loop_outer}.
We denote by $S'_j$ the set $\Edges[u][\CloseParenthesis_i]$, and by $S_j$ the set $S$ constructed in Line~\ref{line:construct_s}.
If $S$ was not constructed in that iteration (i.e., the condition in Line~\ref{line:if_parent} does not hold),
then we let $S_j=\emptyset$.
It is easy to see that $|S_j|\leq |S'_j|$ for all $j$.
The following crucial lemma bounds the total sizes of the sets $S'_j$ constructed throughout the execution of $\BidirectedAlgo$.

\begin{lemma}\label{lem:S_bound}
Let $r$ be the number of iterations of the main loop in Line~\ref{line:loop_outer}.
We have $\sum_{j=1}^{r} |S'_j| = O(m)$.
\end{lemma}
\begin{proof}

By Lemma~\ref{lem:outer_loop_bound} we have $r=O(n)$.
Let $J=\{j: |S'_j|\geq 2 \}$, and it suffices to prove that $\sum_{j\in J} S'_j = O(m)$.

We first argue that after a pair $(u,\CloseParenthesis_i)$ has been extracted from $\Queue$ in some iteration $j\in J$,
the number of edges in $\Edges$ decreases by at least $|S'_j|-1$.
We consider the following complementary cases depending on the condition of Line~\ref{line:if_update}.
\begin{compactenum}
\item If the condition holds, then we have $|\Edges[u][\CloseParenthesis_i]|= 1$ after Line~\ref{line:new_edges} has been executed.
\item Otherwise, we must have $u\in S$ and $|S|\geq 2$, hence there exists some $x\in S\setminus \{u\}$ chosen in Line~\ref{line:choose_x}, and all edges in $\Edges[u]$ will be moved to $\Edges[x]$ for some $v=u$ in Line~\ref{line:loop_inner}.
Hence $|\Edges[u][\CloseParenthesis_i]| = 0$.
\end{compactenum}
Note that because of Line~\ref{line:if_move}, the edges in $\Edges[u][\CloseParenthesis_i]$ are not moved to $\Edges[x][\CloseParenthesis_i]$,
hence all $\Edges[u][\CloseParenthesis_i]$ (except possibly one) will no longer be present at the end of the iteration.
Since $S'_j=\Edges[u][\CloseParenthesis_i]$ at the beginning of the iteration, we obtain that the number of edges in $\Edges$ decreases by at least $|S'_j|-1$.

We define a potential function $\Potential:\Nats \to \Nats$, such that $\Potential(j)$ equals the number of elements in the data structure $\Edges$ at the beginning of the $j$-th iteration of the main loop in Line~\ref{line:loop_outer}.
Note that (i)~initially $\Potential(1)=m$, (ii)~$\Potential(j)\geq 0$ for all $j$, and (iii)~$\Potential(j+1)\leq \Potential(j)$ for all $j$, as new edges are never added to $\Edges$.
Let $(u,\CloseParenthesis_i)$ be an element extracted from $\Queue$ at the beginning of the $j$-the iteration, for some $j\in J$.
As shown above, at the end of the iteration we have removed at least $|S'_j|-1$ edges from $\Edges$, and since $|S'_j| \geq 2$,
we obtain $\Potential(j+1) \leq  \Potential(j) - |S'_j|/2$.
Summing over all $j\in J$, we obtain

\begin{flalign*}
\sum_{j\in J} |S'_j|  &\leq  2\cdot \sum_{j\in J} \left(\Potential(j) - \Potential(j+1) \right) & \left[\text{as $\Potential(j+1) \leq  \Potential(j) - |S'_j|/2$}\right]\\
& = 2\cdot \sum_{\ell=1}^{|J|} \left(\Potential(j_{\ell}) - \Potential(j_{\ell}+1)\right)  & \left[\text{for $j_{\ell}<j_{\ell+1} $}\right]\\
&\leq 2\cdot \Potential(j_1) & \left[ \text{as $\Potential$ is decreasing and thus $\Potential(j_{\ell+1}) \leq \Potential(j_{\ell}+1) $} \right]\\
& \leq 2\cdot m & \left[\text{as $\Potential(j_1) \leq \Potential(1)=m$}\right]
\end{flalign*}

The desired result follows.
\end{proof}

Finally, we are ready to establish the complexity of $\BidirectedAlgo$.

\begin{lemma}[Complexity]\label{lem:bidirected_complexity}
$\BidirectedAlgo$ requires $O(m\cdot\alpha(n))$ time and $O(m)$ space.
\end{lemma}

\noindent{\bf A speedup for non-sparse graphs.}
Observe that in the case of sparse graphs $m=O(n)$, and Lemma~\ref{lem:bidirected_complexity} yields the complexity $O(n\cdot \alpha(n))$.
Here we describe a modification of $\BidirectedAlgo$ that reduces the complexity
from $O(m\cdot \alpha(n))$ to $O(m+n\cdot \alpha(n))$, and thus is faster for graphs
where the edges are more than a factor $\alpha(n)$ as many as the nodes
(i.e., $m=\omega(n\cdot \alpha(n))$).
The key idea is that if a node $u$ has more than $k$ outgoing edges initially,
then it has two distinct outgoing edges labeled with the same closing parenthesis $\CloseParenthesis_i\in \SetClosedParenthesis_k$, and hence the corresponding neighbors can be merged to a single DSCC in a preprocessing step.
Once such a merging has taken place, $u$ only needs to keep a single outgoing edge labeled with $\CloseParenthesis_i$ to that DSCC.
This preprocessing phase requires $O(m)$ time for all nodes, after which there are only $O(n)$ edges present,
by amortizing at most $k$ edges per node of the original graph (recall that $k=O(1)$).
After this preprocessing step has taken place, $\BidirectedAlgo$ is executed with $O(n)$ edges in its input,
and by Lemma~\ref{lem:bidirected_complexity} the complexity is $O(n\cdot \alpha(n))$.
We conclude the results of this section with the following theorem.

\begin{theorem}[Worst-case complexity]\label{them:bidirected_upper}
Let $G=(V,E)$ be a $\Alphabet_k$-labeled bidirected graph of $n$ nodes and $m=\Omega(n)$ edges.
$\BidirectedAlgo$ correctly computes the DSCCs of $G$ and requires $O(m+n\cdot \alpha(n))$ time and $O(m)$ space.
\end{theorem}

\noindent{\bf Linear-time considerations.} Note that $\alpha(n)$ is an extremely slowly growing function, and for all practical purposes $\alpha(n)\leq 4$.
Indeed, the smallest $n$ for which $\alpha(n)=5$ far exceeds the estimated number of atoms in the observable universe.
Additionally, since it is known that a Disjoint Sets data structure operates in amortized constant expected time per operation~\cite{Doyle76,Yao85},
we obtain the following corollary regarding the expected time complexity of our algorithm.

\begin{corollary}[Average-case complexity]\label{cor:bidirected_expected_linear}
For bidirected graphs, the algorithm $\BidirectedAlgo$ requires $O(m)$ expected time for computing DSCCs.
\end{corollary}

%\begin{remark}[Comparison with existing work.]\label{rem:bidirected_comparison}
%We briefly compare our work with the previous best-known results of~\cite{Zhang13}.
%\begin{compactitem}
%\item {\em Algorithm~5.}
%Algorithm~5 of~\cite{Zhang13} solves Dyck reachability on bidirected graphs,
%and the complexity is $O(m\cdot \log n)$ (see~\cite[Theorem~4]{Zhang13}).
%Although not explicitly mentioned in the theorem, the complexity bound is for the 
%\emph{average-case complexity}, and not the \emph{worst-case complexity}.
%The average case comes from their Fast-Doubly-Linked-List (FDLL) data structure,
%the query and deletion time of which are taken to be $O(1)$ in the average case.
%However, the worst-case time complexity of each of these is $O(n)$.
%Using the worst-case time for each operation in FDLL yields the upper bound of 
%$O(n\cdot m\cdot log n)$ time, since (i)~as argued in~\cite[Theorem~4]{Zhang13} the main 
%loop of the algorithm is executed $O(m\cdot \log n)$ times,
%and (ii)~every FDLL query and delete operation inside that loop takes $O(n)$ time instead of $O(1)$.
%
%\item {\em Algorithm~2.}
%The same work presents Algorithm~2 which has worst-case complexity $O(n^2)$
%(shown in~\cite[Theorem~2]{Zhang13}), and thus dominates $O(n\cdot m\cdot\log n)$ 
%on graphs with no isolated nodes.
%\end{compactitem}
%Hence, until now, the best worst-case complexity for the problem has been $O(n^2)$, 
%and the best average-case complexity has been $O(\min\{ n^2,m\cdot \log n\})$.
%The new bounds we establish are $O(m+n\cdot \alpha(n))$ and $O(m)$, respectively.
%\end{remark}

\subsection{An $\Omega(m+n\cdot \alpha(n))$ Lower Bound}\label{subsec:bidirected_lower}
Theorem~\ref{them:bidirected_upper} implies that Dyck reachability on bidirected graphs 
can be solved in almost-linear time.
A theoretically interesting question is whether the problem can be solved in linear time in the 
worst case.
We answer this question in the negative by proving that every algorithm 
for the problem requires $\Omega(m+n\cdot \alpha(n))$ time, and
thereby proving that our algorithm $\BidirectedAlgo$ is indeed optimal wrt worst-case complexity.

\noindent{\bf The Separated Union-Find problem.}
A sequence $\sigma$ of $\Union$-$\Find$ operations is called \emph{separated}
if all $\Find$ operations occur at the end of $\sigma$.
Hence $\sigma=\sigma_1\circ \sigma_2$, where $\sigma_1$ contains all $\Union$ operations of $\sigma$.
We call $\sigma_1$ a \emph{union sequence} and $\sigma_2$ a \emph{find sequence}.
The \emph{Separated Union-Find} problem is the regular Union-Find problem over separated union-find sequences.
Note that this version of the problem has an \emph{offline} flavor, as, at the time when the algorithm is needed to produce output (i.e. when the suffix of $\Find$ operations starts) the input has been fixed (i.e., all $\Union$ operations are known).
We note that the Separated Union-Find problem is different from the Static Tree Set Union problem~\cite{Gabow85},
which restricts the type of allowed $\Union$ operations, and for which a linear time algorithm exists on the RAM model.
The following lemma states a lower bound on the worst-case complexity of the problem.

\begin{lemma}\label{lem:separated_union_find_lowerbound}
The Separated Union-Find problem over a universe of size $n$ and sequences of length $m$ has worst-case complexity $\Omega(m\cdot \alpha(n))$.
\end{lemma}
\begin{proof}
The proof is essentially the proof of \cite[Theorem~4.4]{Tarjan79}, by observing that the sequences constructed there
to prove the lower bound are actually separated union-find sequences.
\end{proof}

\noindent{\bf The union graph $G^{\sigma_1}$.}
Let $\sigma_1$ be a union sequence over some universe $X$.
The \emph{union graph} of $\sigma_1$ is a $\Alphabet_1$-labeled bidirected graph $G^{\sigma_1}=(V^{\sigma_1},E^{\sigma_1})$, defined as follows
(see Figure~\ref{fig:bidirected_hardness} for an illustration).
\begin{compactenum}
\item The node set is $V^{\sigma_1}=X\cup \{z_i\}_{1\leq i\leq |\sigma_1|}$ where the nodes $z_i$ do not appear in $X$.
\item The edge set is $E^{\sigma_1}=\{(z_i, x_i, \CloseParenthesis), (z_i, y_i, \CloseParenthesis)\}_{1\leq i \leq |\sigma_1|}$, where $x_i,y_i\in X$ are the elements such that the $i$-th operation of $\sigma_1$ is $\Union(x_i, y_i)$.
\end{compactenum}

\begin{figure}[!h]
\centering
\begin{tikzpicture}[thick, >=latex,
pre/.style={<-,shorten >= 1pt, shorten <=1pt, thick},
post/.style={->,shorten >= 1pt, shorten <=1pt,  thick},
und/.style={very thick, draw=gray},
node1/.style={circle, minimum size=3.5mm, draw=black!100, line width=1pt, inner sep=0},
node2/.style={circle, minimum size=5mm, draw=black!100, fill=white!100, very thick, inner sep=0},
virt/.style={circle,draw=black!50,fill=black!20, opacity=0}]

\newcommand{\xdisposition}{-2.5}
\newcommand{\ydisposition}{0}
\newcommand{\xstep}{1.5}
\newcommand{\ystep}{1.1}

\node[] at (0, -1)	{$\sigma_1=\Union(u,v),~\Union(x,y),~\Union(w,v),~\Union(w,x)$};

\node[node2] (u) at (\xdisposition + 0*\xstep, \ydisposition + 0*\ystep) {$u$};
\node[node2] (v) at (\xdisposition + 1*\xstep, \ydisposition + 0*\ystep) {$v$};

\node[node2] (x) at (\xdisposition + 3*\xstep, \ydisposition + 0*\ystep) {$x$};
\node[node2] (y) at (\xdisposition + 4*\xstep, \ydisposition + 0*\ystep) {$y$};

\node[node2] (w) at (\xdisposition + 2*\xstep, \ydisposition + 0*\ystep) {$w$};

\node[node2] (z1) at (\xdisposition + 0.5*\xstep, \ydisposition + 1*\ystep) {$z_1$};
\node[node2] (z2) at (\xdisposition + 3.5*\xstep, \ydisposition + 1*\ystep) {$z_2$};
\node[node2] (z3) at (\xdisposition + 1.5*\xstep, \ydisposition + 1*\ystep) {$z_3$};
\node[node2] (z4) at (\xdisposition + 2.5*\xstep, \ydisposition + 1*\ystep) {$z_4$};

\draw[->, very thick, bend right=10] (z1) to node[left, pos=0.2] {$\CloseParenthesis$} (u);
\draw[->, very thick, bend left=10] (z1) to node[right, pos=0.2] {$\CloseParenthesis$} (v);

\draw[->, very thick, bend right=10] (z2) to node[left, pos=0.2] {$\CloseParenthesis$} (x);
\draw[->, very thick, bend left=10] (z2) to node[right, pos=0.2] {$\CloseParenthesis$} (y);

\draw[->, very thick, bend right=10] (z3) to node[left, pos=0.2] {$\CloseParenthesis$} (v);
\draw[->, very thick, bend left=10] (z3) to node[right, pos=0.2] {$\CloseParenthesis$} (w);

\draw[->, very thick, bend right=10] (z4) to node[left, pos=0.2] {$\CloseParenthesis$} (w);
\draw[->, very thick, bend left=10] (z4) to node[right, pos=0.2] {$\CloseParenthesis$} (x);

\end{tikzpicture}
\caption{A union sequence $\sigma_1$ and the corresponding graph $G^{\sigma_1}$.}\label{fig:bidirected_hardness}
\end{figure}

\noindent{\bf A lower bound for Dyck reachability on bidirected graphs.}
We are now ready to prove our lower bound.
The proof consists in showing that there exists no algorithm that solves the problem in $o(m\cdot \alpha(n))$ time.
Assume towards contradiction otherwise, and let $A'$ be an algorithm that solves the problem in time $o(m\cdot \alpha(n))$.
We construct an algorithm $A$ that solves the Separated Union-Find problem in the same time.

Let $\sigma=\sigma_1\circ \sigma_2$ be a separated union-find sequence, where $\sigma_1$ is a union sequence and $\sigma_2$ is a find sequence.
The algorithm $A$ operates as follows. It performs no operations until the whole of $\sigma_1$ has been revealed. 
Then, $A'$ constructs the union graph $G^{\sigma_1}$, and uses $A'$ to solve the 
Dyck reachability problem on $G^{\sigma_1}$.
Finally, every $\Find(x)$ operation encountered in $\sigma_2$ is handled by $A$ by using the answer of $A'$ on $G^{\sigma_1}$.

It is easy to see that $A$ handles the input sequence $\sigma$ correctly.
Indeed,  for any sequence of union operations $\Union(x_i, y_i),\dots, \Union(x_j, y_j)$ that bring two elements $x$ and $y$ to the same set,
the edges $(z_i,x_i,\CloseParenthesis), (z_i,y_i,\CloseParenthesis), \dots,  (z_j,x_j,\CloseParenthesis), (z_j,y_j,\CloseParenthesis)$
must bring $x$ and $y$ to the same DSCC of $G^{\Sigma_1}$.
Finally, the algorithm $A$ requires $O(m)$ time for constructing $G$ and answering all queries, plus $o(m\cdot \alpha(n))$ time for running $A'$ on $G^{\Sigma_1}$.
Hence $A$ operates in $o(m\cdot \alpha(n))$ time, which contradicts Lemma~\ref{lem:separated_union_find_lowerbound}.

We have thus arrived at the following theorem.

\begin{theorem}[Lower-bound]\label{them:bidirected_lower}
Any Dyck reachability algorithm for bidirected graphs with $n$ nodes and 
$m=\Omega(n)$ edges requires $\Omega(m+n\cdot \alpha(n))$ time in the worst case.
\end{theorem}

Theorem~\ref{them:bidirected_lower} together with Theorem~\ref{them:bidirected_upper} 
yield the following corollary.

\begin{corollary}[Optimality]\label{cor:bidirected_optimal}
The Dyck reachability algorithm $\BidirectedAlgo$ for bidirected graphs is optimal wrt to worst-case complexity.
\end{corollary}

\section{Dyck Reachability on General Graphs}\label{sec:dyck_hardness}
In this section we present a hardness result regarding the Dyck reachability problem on general graphs, as well as on graphs of constant treewidth.

%Recall Dyck languages are a subset of Context-free languages.
%The problem of language parsing of a string of length $n$ 
%is a special case of language reachability on a graph of $n+1$ 
%nodes arranged in a line.

\noindent{\bf Complexity of Dyck reachability.} 
Dyck reachability on general graphs is one of the most standard algorithmic formulations of various static analyses.
The problem is well-known to admit a cubic-time solution, while the currently best bound is $O(n^3/\log n)$ due to~\cite{Chaudhuri08}.
Dyck reachability is also known to be 2NPDA-hard~\cite{Heintze97}, which yields a conditional cubic lower bound wrt polynomial improvements.
Here we investigate further the complexity of Dyck reachability. 
We prove that Dych reachability is Boolean Matrix Multiplication (BMM)-hard.
Note that since Dyck reachability is a combinatorial graph problem, techniques such as 
fast-matrix multiplication (e.g. Strassen's algorithm~\cite{Strassen69}) are unlikely to be applicable.
Hence we consider combinatorial (i.e., discrete, graph-theoretic) algorithms.
The standard BMM-conjecture~\cite{Lee02,HKNS15,WilliamsW10,AbboudW14} states that there is no 
truly sub-cubic ($O(n^{3-\delta})$, for $\delta>0$) combinatorial algorithm for Boolean Matrix Multiplication.
Given this conjecture, various algorithmic works establish conditional hardness results.
Here we show that Dyck reachability is BMM-hard on general graphs, which yields a new conditional cubic lower bound for the problem.
Additionally, we show that BMM hardness also holds for Dyck reachability on graphs of constant treewidth.
We establish this by showing Dyck reachability on general graphs is hard as CFL parsing.

%\noindent{\bf Theoretical question.} 
%In parsing, there is a big difference between Dyck and general CFL languages.
%CFL parsing is known to  Boolean Matrix Multiplication hard~\cite{Lee02}, whereas 
%Dyck parsing can be easily solved in $O(n)$ time.
%
%Given the linear-time algorithm for Dyck parsing, an important theoretical question is whether Dyck reachability for 
%general graphs can be solved in truly sub-cubic time, since none of the existing algorithms
%is truly sub-cubic.
%Note that since Dyck reachability is a combinatorial graph problem, techniques such as 
%fast-matrix multiplication (e.g. Strassen's algorithm~\cite{Strassen69}) are unlikely to be 
%applicable.
%Hence we consider combinatorial (i.e., discrete, graph-theoretic) algorithms.
%The standard BMM-conjecture~\cite{Lee02,HKNS15,WilliamsW10,AbboudW14} states that there is no 
%truly sub-cubic ($O(n^{3-\delta})$, for $\delta>0$) combinatorial algorithm for Boolean Matrix Multiplication.
%%Given this conjecture, various algorithmic works establish conditional hardness results.
%We resolve the question for Dyck reachability on general graphs in negative, 
%under the BMM-conjecture, that is, we show that Dyck reachability on general graphs is BMM-hard.
%We establish this by showing Dyck reachability on general graphs 
%is hard as CFL parsing, which we present below.

\begin{figure}[!h]
\begin{subfigure}[b]{.3\textwidth}
\centering
\begin{tikzpicture}[thick, >=latex,
pre/.style={<-,shorten >= 1pt, shorten <=1pt, thick},
post/.style={->,shorten >= 1pt, shorten <=1pt,  thick},
und/.style={very thick, draw=gray},
node1/.style={circle, minimum size=3.5mm, draw=black!100, line width=1pt, inner sep=0},
node2/.style={circle, minimum size=5mm, draw=black!100, fill=white!100, very thick, inner sep=0},
gadget/.style={rectangle, minimum width=8mm, minimum height=11mm, draw=black!70, fill=white!100, very thick, inner sep=0, dashed},
virt/.style={circle,draw=black!50,fill=black!20, opacity=0}]

\newcommand{\xdisposition}{5}
\newcommand{\ydisposition}{0}
\newcommand{\xstep}{1.3}
\newcommand{\ystep}{0.9}

\node[] at (-1.5,0.5) {
$
\begin{aligned}
\StartNonTerminal &\to  \mathcal{T} ~ \mathcal{B}\\
\mathcal{T} &\to \mathcal{A} ~ \StartNonTerminal\\
\mathcal{A} &\to  a\\
\mathcal{B} &\to  b
\end{aligned}
$
};
\end{tikzpicture}
\caption{}\label{subfig:grammar}
\end{subfigure}
\begin{subfigure}[b]{.6\textwidth}
\centering
\begin{tikzpicture}[thick, >=latex,
pre/.style={<-,shorten >= 1pt, shorten <=1pt, thick},
post/.style={->,shorten >= 1pt, shorten <=1pt,  thick},
und/.style={very thick, draw=gray},
node1/.style={circle, minimum size=3.5mm, draw=black!100, line width=1pt, inner sep=0},
node2/.style={circle, minimum size=5mm, draw=black!100, fill=white!100, very thick, inner sep=0},
gadget/.style={rectangle, minimum width=8mm, minimum height=11mm, draw=black!70, fill=white!100, very thick, inner sep=0, dashed},
virt/.style={circle,draw=black!50,fill=black!20, opacity=0}]

\newcommand{\xdisposition}{5}
\newcommand{\ydisposition}{0}
\newcommand{\xstep}{1.3}
\newcommand{\ystep}{0.9}

\node	[node2]		(x)	at	(\xdisposition + 3*\xstep, \ydisposition + 0.5*\ystep)	{$x$};
\node	[node2]		(y)	at	(\xdisposition + -2*\xstep, \ydisposition + 0.5*\ystep)	{$y$};
\node	[node2]		(S)	at	(\xdisposition + 1*\xstep, \ydisposition + 2*\ystep)	{$\StartNonTerminal$};
\node	[node2]		(T)	at	(\xdisposition + 1*\xstep, \ydisposition + 1*\ystep)	{$\mathcal{T}$};
\node	[node2]		(A)	at	(\xdisposition + 1*\xstep, \ydisposition - 0*\ystep)	{$\mathcal{A}$};
\node	[node2]		(B)	at	(\xdisposition + 1*\xstep, \ydisposition - 1*\ystep)	{$\mathcal{B}$};

\node	[node2]		(TT)	at	(\xdisposition + -0.5*\xstep, \ydisposition + 0.75*\ystep)	{};

\node	[node2]		(SS)	at	(\xdisposition + -0.5*\xstep, \ydisposition + 1.75*\ystep)	{};

\draw[->, thick, bend right=20] (x) to node[above] {$\CloseParenthesis_{\StartNonTerminal}$} (S);
\draw[->, thick, bend right=10] (x) to node[above] {$\CloseParenthesis_{\mathcal{T}}$} (T);
\draw[->, thick, bend left=10] (x) to node[below] {$\CloseParenthesis_{\mathcal{A}}$} (A);
\draw[->, thick, bend left=20] (x) to node[below] {$\CloseParenthesis_{\mathcal{B}}$} (B);

\draw[->, thick, bend left=20] (B) to node[below] {$\OpenParenthesis_{b}$} (y);
\draw[->, thick, bend left=10] (A) to node[below] {$\OpenParenthesis_{a}$} (y);

\draw[->, thick] (T) to node[above] {$\OpenParenthesis_{\mathcal{\StartNonTerminal}}$} (TT);
\draw[->, thick] (TT) to node[above] {$\OpenParenthesis_{A}$} (y);

\draw[->, thick] (S) to node[above] {$\OpenParenthesis_{\mathcal{B}}$} (SS);
\draw[->, thick, bend right=10] (SS) to node[above] {$\OpenParenthesis_{\mathcal{T}}$} (y);

\end{tikzpicture}
\caption{}\label{subfig:gadget}
\end{subfigure}

\vspace{0.3cm}

\begin{subfigure}[b]{\textwidth}
\centering
\begin{tikzpicture}[thick, >=latex,
pre/.style={<-,shorten >= 1pt, shorten <=1pt, thick},
post/.style={->,shorten >= 1pt, shorten <=1pt,  thick},
und/.style={very thick, draw=gray},
node1/.style={circle, minimum size=7mm, draw=black!100, fill=white!100, very thick, inner sep=0},
node2/.style={circle, minimum size=7mm, draw=black!100, fill=white!100, very thick, inner sep=0},
gadget/.style={rectangle, minimum width=9mm, minimum height=11mm, draw=black!70, fill=white!100, very thick, inner sep=0, dashed},
virt/.style={circle,draw=black!50,fill=black!20, opacity=0}]

\newcommand{\xdisposition}{0}
\newcommand{\ydisposition}{-4.8}
\newcommand{\xstep}{1.5}
\newcommand{\ystep}{0.9}

\node[node2] (0) at	(\xdisposition + -1*\xstep, \ydisposition + 0.5*\ystep)	{$v$};
\node[node2] (1) at	(\xdisposition + 0*\xstep, \ydisposition + 0.5*\ystep)	{$u_0$};
\node[node2] (2) at	(\xdisposition +1*\xstep, \ydisposition + 0.5*\ystep)	{$u_1$};
\node[node2] (3) at	(\xdisposition +2*\xstep, \ydisposition + 0.5*\ystep)	{$u_2$};
\node[] (ax1) at (\xdisposition +3*\xstep, \ydisposition + 0.5*\ystep)	{};
\node[] at (\xdisposition +4*\xstep, \ydisposition + 0.5*\ystep)	{$\cdots$};
\node[node2] (n-1) at	(\xdisposition +5*\xstep, \ydisposition + 0.5*\ystep)	{$u_{n-1}$};
\node[node2] (n) at	(\xdisposition +6*\xstep, \ydisposition + 0.5*\ystep)	{$u_n$};

\draw[->, thick] (0) to node[below] {$\OpenParenthesis_{\StartNonTerminal}$} (1);
\draw[->, thick] (1) to node[below] {$\CloseParenthesis_{s_1}$} (2);
\draw[->, thick] (2) to node[below] {$\CloseParenthesis_{s_2}$} (3);
\draw[->, thick] (3) to node[below] {$\CloseParenthesis_{s_3}$} (ax1);
%\draw[->, thick] (ax2) to node[below] {$\CloseParenthesis_{s_{n-1}}$} (n-1);
\draw[->, thick] (n-1) to node[below] {$\CloseParenthesis_{s_n}$} (n);

\node[gadget] (g0) at	(\xdisposition + 0*\xstep, \ydisposition + 1.8*\ystep+0.53)	{};
\node[node1] (x0) at	(\xdisposition + 0*\xstep, \ydisposition + 1.8*\ystep)	{$x_0$};
\node[node1] (y0) at	(\xdisposition + 0*\xstep, \ydisposition + 1.8*\ystep+1.08)	{$y_0$};
\draw[->, thick] (1) to (x0);
\draw[->, thick, out=150, in=140] (y0) to (1);

\node[gadget] (g1) at	(\xdisposition + 1*\xstep, \ydisposition + 1.8*\ystep+0.53)	{};
\node[node1] (x1) at	(\xdisposition + 1*\xstep, \ydisposition + 1.8*\ystep)	{$x_1$};
\node[node1] (y1) at	(\xdisposition + 1*\xstep, \ydisposition + 1.8*\ystep+1.08)	{$y_1$};
\draw[->, thick] (2) to (x1);
\draw[->, thick, out=150, in=140] (y1) to (2);

\node[gadget] (g2) at	(\xdisposition + 2*\xstep, \ydisposition + 1.8*\ystep+0.53)	{};
\node[node1] (x2) at	(\xdisposition + 2*\xstep, \ydisposition + 1.8*\ystep)	{$x_2$};
\node[node1] (y2) at	(\xdisposition + 2*\xstep, \ydisposition + 1.8*\ystep+1.08)	{$y_2$};
\draw[->, thick] (3) to (x2);
\draw[->, thick, out=150, in=140] (y2) to (3);

\node[gadget] (gn-1) at	(\xdisposition + 5*\xstep, \ydisposition + 1.8*\ystep+0.53)	{};
\node[node1] (xn-1) at	(\xdisposition + 5*\xstep, \ydisposition + 1.8*\ystep)	{$x_{n-1}$};
\node[node1] (yn-1) at	(\xdisposition + 5*\xstep, \ydisposition + 1.8*\ystep+1.08)	{$y_{n-1}$};
\draw[->, thick] (n-1) to (xn-1);
\draw[->, thick, out=150, in=140] (yn-1) to (n-1);

\end{tikzpicture}
\caption{}\label{subfig:parse_graph}
\end{subfigure}
\caption{\textbf{(\ref{subfig:grammar})} A grammar $\Grammar$ for the language $a^nb^n$, \textbf{(\ref{subfig:gadget})} The gadget graph $G^{\Grammar}$, \textbf{(\ref{subfig:parse_graph})} The parse graph $G_{s}^{\Grammar}$, given a string $s=s_1,\dots d_n$.}\label{fig:parsing_to_dyck}
\end{figure}
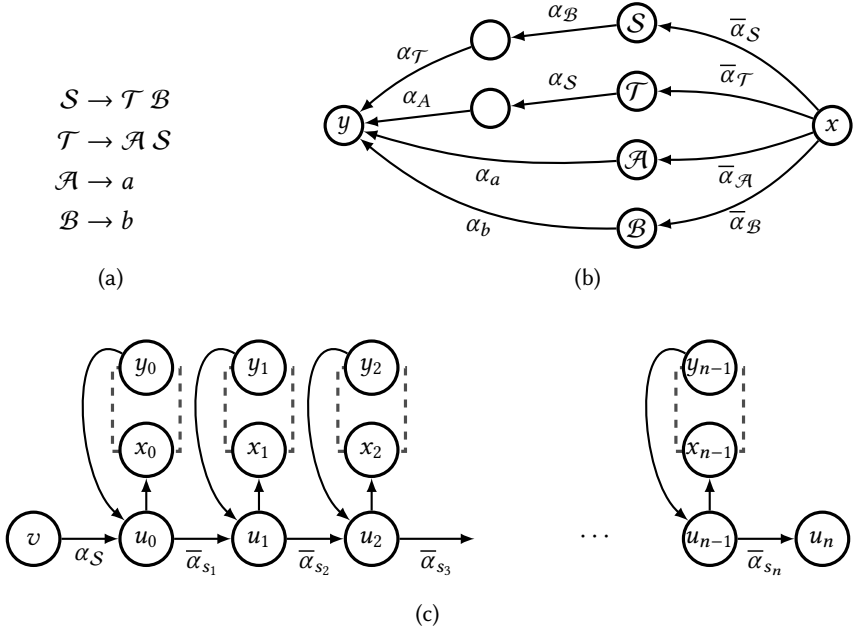

\noindent{\bf The gadget graph $G^{\Grammar}$.}
Given a Context-free grammar $\Grammar$ in Chomsky normal form, we construct the \emph{gadget graph} $G^{\Grammar}=(V^{\Grammar}, E^{\Grammar})$ as follows (see Figure~\ref{fig:parsing_to_dyck} (\ref{subfig:grammar}), (\ref{subfig:gadget}) for an illustration).
\begin{compactenum}
\item The node set $V^{\Grammar}$ contains two distinguished nodes $x,y$, together with a node $x_i$ for the $i$-th production rule $p_i$. Additionally, if $p_i$ is of the form $\mathcal{A}\to \mathcal{B}~\mathcal{C}$, then $V^{\Grammar}$ contains a node $y_i$.
\item The edge set $E^{\Grammar}$ contains an edge $(x, x_i, \CloseParenthesis_{\mathcal{A}})$, where $\mathcal{A}$ is the left hand side symbol of the $i$-th production rule $p_i$ of $\Grammar$.
Additionally, 
\begin{compactenum}
\item if $p_i$ is of the form $\mathcal{A} \to a$, then $E^{\Grammar}$ contains an edge $(x_i, y, \OpenParenthesis_a)$, else
\item if $p_i$ is of the form $\mathcal{A} \to \mathcal{B}~\mathcal{C}$, then $E^{\Grammar}$ contains the edges $(x_i, y_i, \OpenParenthesis_{\mathcal{C}})$ and $(y_i, y, \OpenParenthesis_{\mathcal{B}})$.
\end{compactenum}
\end{compactenum}

\noindent{\bf The parse graph $G^{\Grammar}_s$.}
Given a grammar $\Grammar$ and an input string $s=s_1,\dots s_n$,
we construct the \emph{parse graph} $G^{\Grammar}_s=(V^{\Grammar}_s, E^{\Grammar}_s)$ as follows.
The graph consists of two parts.
The first part is a line graph that contains nodes $v, u_0, u_1,\dots u_n$,
with the edges $(v, u_0, \OpenParenthesis_{\StartNonTerminal})$ and $(u_{i-1}, u_{i}, \CloseParenthesis_{s_i} )$ for all $1\leq i\leq n$.
The second part consists of a $n$ copies of the gadget graph $G^{\Grammar}$, counting from $0$ to $n-1$.
Finally, we have a pair of edges $(u_i, x_i, \epsilon)$, $(y_i, u_i, \epsilon)$ for every $0\leq i < n$,
where $x_i$ (resp. $y_i$) is the distinguished $x$ node (resp. $y$ node) of the $i$-th gadget graph.
See Figure~\ref{fig:parsing_to_dyck} (\ref{subfig:parse_graph}) for an illustration.

\begin{lemma}\label{lem:parse_graph_correctness}
The node $u_n$ is Dyck-reachable from node $v$ iff $s$ is generated by $\Grammar$.
\end{lemma}
\begin{proof}
Given a path $P$, we denote by $\ClosingLabel(P)$ the substring of $\Label(P)$ that consists of all the closing-parenthesis symbols of $\Label(P)$.
The proof follows directly from the following observation:
the parse graph $G^{\Grammar}_s$ contains a path $P:v\Path u_n$ with $\Label(P)\in \Dyck$ if and only if $\ClosingLabel(P)$ corresponds to a pre-order traversal of a derivation tree of the string $s$ wrt the grammar $\Grammar$.
\end{proof}

\begin{theorem}\label{them:dyck_hard}
If there exists a combinatorial algorithm that solves the pair Dyck reachability problem in time $\mathcal{T}(n)$, where $n$ is the number of nodes of the input graph, then there exists a combinatorial algorithm that solves the CFL parsing problem in time $O(n+\mathcal{T}(n))$.
\end{theorem}

%Note that due to Remark~\ref{rem:parse_graph_treewidth}, the above theorem holds even if we restrict the Dyck reachability problem to graphs of constant treewidth.
Since CFL-parsing is BMM-hard, by combining Theorem~\ref{them:dyck_hard} with~\cite[Theorem~2]{Lee02}
we obtain the following corollary.

\begin{corollary}[BMM-hardness: Conditional cubic lower bound]\label{cor:dyck_hard}
For any fixed $\delta>0$, if there is a combinatorial algorithm that solves the pair Dyck reachability problem in $O(n^{3-\delta})$ time,
then there is a combinatorial algorithm that solves Boolean Matrix Multiplication in $O(n^{3-\delta/3})$ time.
\end{corollary}

\begin{remark}[\bf BMM hardness for low-treewidth graphs]\label{rem:dyck_hard_tw}
Note that since the size of the grammar $\Grammar$ is constant, the parse graph $G^{\Grammar}_s$ has constant treewidth.
Hence the BMM hardness of Corollary~\ref{cor:dyck_hard} also holds if we restrict our attention to Dyck reachability on graphs of constant treewidth.
\end{remark}

%\noindent{\bf Implications.}
%We briefly discuss the implications of Corollary~\ref{cor:dyck_hard}.
%Over the last 50 years, there have been considerable efforts in devising fast matrix multiplication algorithms,
%starting with the seminal work of~\cite{Strassen69}, 
%and with the current best bound being $O(n^{\omega})$ for $\omega=2.3728639$ due to~\cite{LeGall14}.
%However, the constants involved in these complexities are large, so that these algorithms are considered impractical (possibly with the exception of Strassen's algorithm).
%Hence, Corollary~\ref{cor:dyck_hard} implies that it is unlikely to exist a simple, practical algorithm for the Dyck reachability problem that runs in time $O(n^{3-\delta})$, for any fixed $\delta>0$.
%Hence, the existing algorithms that operate in $O(n^3)$ time are optimal.

\section{Library/Client Dyck Reachability}\label{sec:library}

In this section we present some new results for library/client Dyck reachability with applications to context-sensitive data-dependence analysis.
One crucial step to our improvements is the fact that we consider that the underlying graphs are not arbitrary, but have special structure.
We start with Section~\ref{subsec:tree_decompositions} which defines formally the graph models we deal with, and their structural properties.
Afterwards, in Section~\ref{subsec:library}, we present our algorithms.

\subsection{Problem Definition}\label{subsec:tree_decompositions}

 Here we present a formal definition of the input graphs that we will be considering for library/client Dyck reachability with application to context-sensitive data-dependence analysis.
Each input graph $G$ is not an arbitrary $\Alphabet_k$-labeled graph, but has two important structural properties.
 \begin{compactenum}
 \item $G$ can be naturally partitioned to subgraphs $G_1,\dots G_{\ell}$, such that every $G_i$ has only $\epsilon$-labeled edges.
Each such $G_i=(V_i, E_i)$ corresponds to a method of the input program. 
There are only few nodes of $V_i$ with \emph{incoming} edges that are non-$\epsilon$-labeled. 
Similarly, there are only few nodes of $V_i$ with \emph{outgoing} edges that are non-$\epsilon$-labeled.
These nodes correspond to the input parameters and return statements of the $i$-th method of the program, which are almost always only a few.
\item Each $G_i$ is a graph of \emph{low treewidth}.
This is an important graph-theoretic property which, informally, means that $G_i$ is similar to a tree (although $G_i$ is not a tree).
 \end{compactenum}

We make the above structural properties formal and precise.
We start with the first structural property, we captures the fact that the input graph $G$ consists of many local graphs $G_i$, one for each method of the input program, and the parenthesis-labeled edges model context sensitivity.

\noindent{\bf Program-valid partitionings.}
Let $G=(V,E)$ be a $\Alphabet_k$-labeled graph.
Given some $1\leq i\leq k$, we define the following sets.

\begin{align*}
&V_c(\OpenParenthesis_i) = \{u: \exists (u,v,\OpenParenthesis_i)\in E\}  \quad
&V_e(\OpenParenthesis_i) = \{v: \exists (u,v,\OpenParenthesis_i)\in E\}\\
&V_x(\CloseParenthesis_i) = \{u: \exists (u,v,\CloseParenthesis_i)\in E\} \quad
&V_r(\CloseParenthesis_i) = \{v: \exists (u,v,\OpenParenthesis_i)\in E\}
\end{align*}
In words, (i)~$V_c(\OpenParenthesis_i)$ contains the nodes that have a $\OpenParenthesis_i$-labeled outgoing edge,
(ii)~$V_e(\OpenParenthesis_i)$ contains the nodes that have a $\OpenParenthesis_i$-labeled incoming edge,
(iii)~$V_x(\CloseParenthesis_i)$ contains the nodes that have a $\CloseParenthesis_i$-labeled outgoing edge,
and (iv)~$V_r(\CloseParenthesis_i)$ contains the nodes that have a $\CloseParenthesis_i$-labeled incoming edge.
Additionally, we define the following sets.
\[
V_c=\bigcup_i V_c(\OpenParenthesis_i) \quad V_e=\bigcup_i V_e(\OpenParenthesis_i) \quad V_x=\bigcup_i V_x(\CloseParenthesis_i) \quad V_r=\bigcup_i V_r(\CloseParenthesis_i)
\]

Consider a partitioning $\Partition=\{V_1,\dots, V_{\ell}\}$ of the node set $V$,
i.e., $\bigcup_iV_i = V$ and $V_i\cap V_j=\emptyset$ for all $1\leq i,j \leq \ell$.
We say that $\Partition$ is \emph{program-valid} if the following conditions hold:
for every $1\leq i\leq k$, there exist some $1\leq j_1,j_2\leq \ell$ such that
(i)~$V_c(\OpenParenthesis_i), V_r(\CloseParenthesis_i)\subseteq V_{j_1}$, and
(ii)~$V_e(\OpenParenthesis_i), V_x(\CloseParenthesis_i)\subseteq V_{j_2}$.
Intuitively, the parenthesis-labeled edges of $G$ correspond to method calls and returns, and thus model context sensitivity.
Each parenthesis type models the calling context, and each $G\restr{V_i}$ corresponds to a single method of the program.
Since the calling context is tied to two methods (the caller and the callee), conditions (i) and (ii) must hold for the partitioning.

A program-valid partitioning $\Partition=\{V_1,\dots V_{\ell}\}$ is called $b$-bounded if there exists some $b\in \Nats$ such that for all $1\leq j \leq \ell$  we have that $|V_e\cap V_j|, |V_x\cap V_j|\leq b$.
Note that since $\Partition$ is program-valid, this condition also yields that for all $1\leq i\leq k$ we have that $|V_c(\OpenParenthesis_i)|, |V_r(\CloseParenthesis_i)|\leq b$. 
In this paper we consider that $b=O(1)$, i.e., $b$ is constant wrt the size of the input graph.
This is true since the sets $V_e\cap V_j$ and $V_x\cap V_j$ represent the input parameters and the return statements of the $j$-th method in the program.
Similarly, the sets $V_c(\OpenParenthesis_i)$, $V_r(\CloseParenthesis_i)$ represent the variables that are passed as input and the variables that capture the return, respectively, of the method that the $i$-th call site refers to.
In all practical cases each of the above sets has constant size (or even size $1$, for return variables).

\noindent{\bf Program-valid graphs.}
The graph $G$ is called \emph{program-valid} if there exists a constant $b\in \Nats$ such that $G$ has $b$-bounded program valid partitioning.
Given a such a partitioning $\Partition=\{V_1,\dots ,V_{\ell}\}$, we call each graph $G_i=(V_i,E_i)=G\restr{V_i}$ a \emph{local graph}.
Given a partitioning of $V$ to the library partition $V^1$ and client partition $V^2$, 
$\Partition$ induces a program-valid partitioning on each of the library subgraph $G^1=G\restr{V^1}$ and $G^2=G\restr{V^2}$.
See Figure~\ref{fig:lib_example} for an example.

We now present the second structural property of input graphs that we exploit in this work.
Namely, for a program-valid input graph $G$ with a program-valid partitioning $\Partition=\{V_1,\dots ,V_{\ell}\}$
the local graphs $G_i=G\restr{V_i}$ have \emph{low treewidth}.
It is known that the control-flow graphs (CFGs) of goto-free programs have small treewidth~\cite{Thorup98}.
The local graphs $G_i$ are not CFGs, but rather graphs defined by def-use chains.
As we show in this work (see Section~\ref{subsec:experiments_library}), the local def-use graphs of real-world benchmarks also have small treewidth.
Below, we make the above notions precise.

\noindent{\bf Trees.}
A (rooted) tree $T=(V_T, E_T)$ is an undirected graph with a distinguished node $w$ which is the root
such that there is a unique simple path $P_u^v:u \Path v$ for each pair of nodes $u,v$.
Given a tree $T$ with root $w$, the \emph{level} $\Level(u)$ of a node $u$ is the length of the simple
path $P_u^w$ from $u$ to the root $r$.
Every node in $P_u^w$ is an \emph{ancestor} of $u$.
If $v$ is an ancestor of $u$, then $u$ is a \emph{descendant} of $v$.
For a pair of nodes $u,v\in V_T$, the \emph{lowest common ancestor (LCA)} of $u$ and $v$
is the common ancestor of $u$ and $v$ with the largest level.
The \emph{parent} $u$ of $v$ is the unique ancestor of $v$ in level $\Level(v)-1$,
and $v$ is a \emph{child} of $u$. A \emph{leaf} of $T$ is a node with no children.
For a node $u\in V_T$, we denote by $T(u)$ the subtree of $T$ rooted in $u$ 
(i.e., the tree consisting of all descendants of $u$). 
%A tree is called \emph{$k$-ary} if every node has at most $k$ children 
%(e.g., in a binary tree every node has at most two children).
%A \emph{full $k$-ary tree} is a $k$-ary tree in which every non-leaf node has exactly $k$ children.
The \emph{height} of $T$ is $\max_u\Level(u)$.
% (i.e., it is the maximum level of its nodes).

\noindent{\bf Tree decompositions and treewidth}~\cite{Robertson84}.
Given a graph $G$, a tree-decomposition $\Tree(G)=(V_T, E_T)$ is a tree with the following properties.
\begin{compactenum}
\item[\emph{C1}:]\label{item:c1} $V_T=\{\Bag_1,\dots, \Bag_{b}: \text{ for all } 
1\leq i \leq b. \ \Bag_i\subseteq V\}$ and $\bigcup_{\Bag_i\in V_T}\Bag_i=V$.
That is, each node of $\Tree(G)$ is a subset of nodes of $G$, and each node of $G$ appears in some node of $\Tree(G)$.
\item[\emph{C2}:]\label{item:c2} For all $(u,v)\in E$ there exists $\Bag_i\in V_T$ such that $u,v\in \Bag_i$.
That is, the endpoints of each edge of $G$ appear together in some node of $\Tree(G)$.
\item[\emph{C3}:]\label{item:c3} For all $\Bag_i$, $\Bag_j$ and any bag $\Bag_k$ that appears in the simple path $\Bag_i\Path \Bag_j$ in $\Tree(G)$,
we have $\Bag_i\cap \Bag_j\subseteq \Bag_k$.
That is, every node of $G$ is contained in a contiguous subtree of $\Tree(G)$.
\end{compactenum}
To distinguish between the nodes of $G$ and the nodes of $\Tree(G)$,
the sets $\Bag_i$ are called \emph{bags}.
The {\em width} of a tree-decomposition $\Tree(G)$ is the size of the largest bag minus~1
and the {\em treewidth} of $G$ is the width of a minimum-width tree decomposition of $G$.
It follows from the definition that if $G$ has constant treewidth, then $m=O(n)$.
%A graph has treewidth $1$ precisely if it is a tree.
For a node $u\in V$, we say that a bag $\Bag$ is the \emph{root bag} of $u$ if $\Bag$ 
is the bag with the smallest level among all bags that contain $u$,
i.e., $\Bag_u=\arg\min_{\Bag\in V_T:~u\in \Bag}\Level\left(\Bag\right)$. 
By definition, there is exactly one root bag for each node $u$. 
We often write $\Bag_u$ for the root bag of node $u$, and denote by $\Level(u)=\Level\left(\Bag_u\right)$.  
Additionally, we denote by $\Bag_{(u,v)}$ the bag of the largest level that is the root bag of one of $u$, $v$.
The following well-known theorem states that tree decompositions of constant-treewidth graphs can be constructed efficiently.

\begin{theorem}[\cite{Bodlaender95}]\label{them:tree_dec}
Given a graph $G=(V,E)$ of $n$ nodes and treewidth $t=O(1)$, a tree decomposition $\Tree(G)$ of $O(n)$ bags, height $O(\log n)$ and width $O(t)=O(1)$ can be constructed in
$O(n)$ time.
\end{theorem}

The following crucial lemma states the key property of tree decompositions that we exploit in this work towards fast algorithms for Dyck reachability.
Intuitively, every bag of a tree decomposition $\Tree(G)$ acts as a separator of the graph $G$.

\begin{lemma}[{\cite[Lemma~3]{Bodlaender98}}]\label{lem:separator_property}
Consider a graph $G=(V,E)$, a tree-decomposition $T=\Tree(G)$, and a bag $\Bag$ of $T$.
Let $(\Comp_i)_{i}$ be the components of $T$ created by removing $\Bag$ from $T$,
and let $V_i$ be the set of nodes that appear in bags of component $\Comp_i$.
For every $i\neq j$, nodes $u\in V_i$, $v\in V_j$ and path $P:u\Path v$, we have that $P\cap\Bag\neq\emptyset$  
(i.e., all paths between $u$ and $v$ go through some node in $\Bag$).
\end{lemma}

\noindent{\bf Program-valid treewidth.}
Let $G=(V,E)$ be a $\Alphabet_k$-labeled program-valid graph, and $\Partition=\{V_1,\dots, V_{\ell}\}$ a program-valid partitioning of $G$.
For each $1\leq i\leq \ell$, let $G_i=(V_i, E_i)=G\restr{V_i}$.
We define the graph $G'_i=(V_i, E'_i)$ such that
\[
E'_i=E_i\bigcup_{1\leq j\leq k}{\left(V_c(\OpenParenthesis_j)\cap V_i\right) \times \left(V_r(\CloseParenthesis_j)\cap V_i\right)}
\]
and call $G'_i$ the \emph{maximal} graph of $G_i$.
In words, the graph $G'_i$ is identical to $G_i$, with the exception that $G'_i$ contains an extra edge for every
pair of nodes $u,v\in V_i$ such that $u$ has opening-parenthesis-labeled outgoing edges, and $v$ has closing-parenthesis-labeled incoming edges.
We define the treewidth of $\Partition$ to be the smallest integer $t$ such that the treewidth of each $G'_i$ is at most $t$.
We define the width of the pair $(G,\Partition)$ as the treewidth of $\Partition$,
and the \emph{program-valid treewidth} of $G$ to be the smallest treewidth among its program-valid partitionings.

\noindent{\bf The Library/Client Dyck reachability problem on program-valid graphs.}
Here we define the algorithmic problem that we solve in this section.
Let $G=(V,E)$ be a $\Alphabet_k$-labeled, program-valid graph and $\Partition$ a program-valid partitioning of $G$ that has constant treewidth
($k$ need not be constant).
The set $\Partition$ is further partitioned into two sets, $\Partition^1$ and $\Partition^2$ that correspond to the 
\emph{library} and \emph{client} partitions, respectively.
We let $V^1=\bigcup_{V_i\in \Partition^1} V_i$ and $V^2=\bigcup_{V_i\in \Partition^2} V_i$,
and define the \emph{library graph} $G^1=(V^1, E^1)=G\restr{V^1}$ and the \emph{client graph} $G^2=(V^2, E^2)=G\restr{V^1}$.

The task is to answer Dyck reachability queries on $G$, where the queries are either
(i)~single source queries from some $u\in V^2$, or (ii)~pair queries for some pair $u,v \in V^2$.
The computation takes place in two phases.
In the \emph{preprocessing phase}, only the library graph $G^1$ is revealed, and we are allowed to some preprocessing to compute reachability summaries.
In the \emph{query phase}, the whole graph $G$ is revealed, and our task is to handle queries fast, by utilizing the preprocessing done on $G^1$.

\begin{figure}
\begin{subfigure}[b]{0.25\textwidth}
\begin{algorithm}[H]
\small
\SetInd{0.4em}{0.4em}
\DontPrintSemicolon
%\setstretch{1.05}
\caption{$f_1(x,y)$}
\BlankLine
\eIf{$y\%2 =1$}{
$z\gets  x + y$
}
{
$z\gets x\cdot y$
}
\Return{$z$}
\end{algorithm}
%\caption{}
%\label{subfig:method_f1}
\end{subfigure}
\quad 
\begin{subfigure}[b]{0.25\textwidth}
\begin{algorithm}[H]
\small
\SetInd{0.4em}{0.4em}
\DontPrintSemicolon
%\setstretch{1.05}
\caption{$g()$}
\BlankLine
$x\gets 2$\\
$y\gets 2$\\
$p\gets f(x,y)$\\
\Return{$p$}
\end{algorithm}
%\caption{}
%\label{subfig:method_g}
\end{subfigure}
\quad
\begin{subfigure}[b]{0.25\textwidth}
\begin{algorithm}[H]
\small
\SetInd{0.4em}{0.4em}
\DontPrintSemicolon
%\setstretch{1.05}
\caption{$f_2(x,y)$}
\BlankLine
\eIf{$x\%2 =1$}{
$z\gets  2\cdot x$
}
{
$z\gets 2\cdot x + 1$
}
\Return{$z$}
\end{algorithm}
%\caption{}
%\label{subfig:method_f2}
\end{subfigure}
\\
\begin{subfigure}{0.8\textwidth}
\small
\centering
\begin{tikzpicture}[thick, >=latex, scale=0.9,
pre/.style={<-,shorten >= 1pt, shorten <=1pt, thick},
post/.style={->,shorten >= 1pt, shorten <=1pt,  thick},
und/.style={very thick, draw=gray},
bag/.style={ellipse, minimum height=7mm,minimum width=12mm,draw=gray!80, line width=1pt, inner sep=0},
internal/.style={very thick ,circle,draw=black!80, inner sep=1, minimum size=4.5mm},
%exit/.style={circle,draw=black!80, inner sep=2, minimum size=4pt},
%call/.style={circle,draw=black!80, inner sep=2, minimum size=4pt},
%return/.style={circle,draw=black!80, inner sep=2, minimum size=4pt},
virt/.style={circle,draw=black!50,fill=black!20, opacity=0}]

\newcommand{\ynodestep}{-0.9}
\newcommand{\xnodestep}{0.6}
\newcommand{\xdisposition}{5}
\newcommand{\xcaptiondisposition}{0.5}
\newcommand{\xtextdisposition}{1}
\newcommand{\legendy}{2.5}

\newcommand{\secondlineybias}{0}

\newcommand{\globalxdisposition}{-3.0}
\newcommand{\globalydisposition}{-7.5}
\renewcommand{\ynodestep}{-0.5}
\renewcommand{\xnodestep}{0.35}

\newcommand{\ybagstep}{-1}
\newcommand{\xbagstep}{0.7}

\node	[internal]		(xmain)	at	(\globalxdisposition+-1.5*\xnodestep,\globalydisposition+1*\ynodestep)		{$1$};
\node	[internal]		(ymain)	at	(\globalxdisposition+1.5*\xnodestep,\globalydisposition+1*\ynodestep)		{$2$};
\node	[internal]		(qmain)	at	(\globalxdisposition+0*\xnodestep,\globalydisposition+3*\ynodestep)		{$3$};
\node	[internal]		(returnmain)	at	(\globalxdisposition,\globalydisposition+4.5*\ynodestep)		{$4$};
% \node	[internal]		(x5)	at	(\globalxdisposition+\xnodestep,\globalydisposition+2*\ynodestep)		{$5$};
% \node	[internal]		(x6)	at	(\globalxdisposition+\xnodestep,\globalydisposition+3*\ynodestep)		{$6$};

\node	[internal]		(xf_)	at	(\globalxdisposition+-12*\xnodestep,\globalydisposition+0*\ynodestep)		{$1$};
\node	[internal]		(yf_)	at	(\globalxdisposition+-9*\xnodestep,\globalydisposition)		                {$2$};
\node	[internal]		(iff_)	at	(\globalxdisposition+-10.5*\xnodestep,\globalydisposition+1.25*\ynodestep)		{$3$};
\node	[internal]		(addf_)	at	(\globalxdisposition-13*\xnodestep,\globalydisposition+3*\ynodestep)		{$4$};
\node	[internal]  (multif_)	at	(\globalxdisposition+-8*\xnodestep,\globalydisposition+3*\ynodestep)		{$5$};
\node	[internal]		(phi_)	at	(\globalxdisposition+-10.5*\xnodestep,\globalydisposition+5*\ynodestep)		{$\phi$};
\node	[internal]		(retf_)	at	(\globalxdisposition+-10.5*\xnodestep,\globalydisposition+6.5*\ynodestep-0.2)		{$6$};

\node	[internal]		(xf__)	at	(\globalxdisposition+9*\xnodestep,\globalydisposition+0*\ynodestep)		{$1$};
\node	[internal]		(yf__)	at	(\globalxdisposition+12*\xnodestep,\globalydisposition)		                {$2$};
\node	[internal]		(iff__)	at	(\globalxdisposition+10.5*\xnodestep,\globalydisposition+1.25*\ynodestep)		{$3$};
\node	[internal]		(addf__)	at	(\globalxdisposition+8*\xnodestep,\globalydisposition+3*\ynodestep)		{$4$};
\node	[internal]  (multif__)	at	(\globalxdisposition+13*\xnodestep,\globalydisposition+3*\ynodestep)		{$5$};
\node	[internal]		(phi__)	at	(\globalxdisposition+10.5*\xnodestep,\globalydisposition+5*\ynodestep)		{$\phi$};
\node	[internal]		(retf__)	at	(\globalxdisposition+10.5*\xnodestep,\globalydisposition+6.5*\ynodestep -0.2)		{$6$};

\draw [->, dashed,thick] (xmain) to 	(qmain);
\draw [->, dashed,thick] (ymain) to 	(qmain);
\draw [->, thick] (qmain) to 	(returnmain);
% \draw [->, thick] (x3) to 	(x4);
% \draw [->, thick, bend right =20] (x4) to node[right]{)} 	(x2);
% \draw [->, thick] (x5) to 	(x6);

\draw [->, thick] (yf_) to 		(iff_);
\draw [->, thick] (xf_) to 		(addf_);
\draw [->, thick, bend left =30 ] (yf_) to 		(addf_);
\draw [->, thick] (yf_) to 		(multif_);
\draw [->, thick,, bend right =30] (xf_) to 		(multif_);
\draw [->, thick] (addf_) to 		(phi_);
\draw [->, thick] (multif_) to 		(phi_);
\draw [->, thick] (phi_) to 		(retf_);

\draw [->, thick] (xf__) to 		(iff__);
\draw [->, thick] (xf__) to 		(addf__);
% \draw [->, thick, bend left =30 ] (yf__) to 		(addf__);
% \draw [->, thick] (yf__) to 		(multif__);
\draw [->, thick, bend right =30] (xf__) to 		(multif__);
\draw [->, thick] (addf__) to 		(phi__);
\draw [->, thick] (multif__) to 		(phi__);
\draw [->, thick] (phi__) to 		(retf__);

\draw [->, thick, bend right =40] (ymain) to 	node[above]{\Large $\}_1$}	(yf_);
\draw [->, thick, bend right =40] (xmain) to 	node[above]{\Large$\}_1$}	(xf_);
\draw [->, thick,bend left =40] (xmain) to 	node[above]{\Large$\}_2$}	(xf__);
\draw [->, thick,bend left =40] (ymain) to 	node[above]{\Large$\}_2$}	(yf__);
\draw [->, thick,bend right =10] (retf_) to node[below]{\Large$\}_1$}		(qmain);
\draw [->, thick,bend left =10] (retf__) to node[below]{\Large$\}_2$}		(qmain);

%3,6 - 5,8 - 5,8
\node[ellipse, minimum width=2cm, minimum height=2.8cm, ultra thick, dashed, draw=gray] at (\globalxdisposition,\globalydisposition+-1.2) {};
\node[ellipse, minimum width=2.9cm, minimum height=4.3cm, ultra thick, dashed, draw=gray] at (\globalxdisposition+3.6,\globalydisposition+-1.5) {};
\node[ellipse, minimum width=2.9cm, minimum height=4.3cm, ultra thick, dashed, draw=gray] at (\globalxdisposition-3.6,\globalydisposition+-1.5) {};

\node[] at (\globalxdisposition-3.6,\globalydisposition+-4.2) {\large $f_1(x,y)$};
\node[] at (\globalxdisposition+3.6,\globalydisposition+-4.2) {\large $f_2(x,y)$};
\node[] at (\globalxdisposition,\globalydisposition+-4.2) {\large $g()$};

\end{tikzpicture}

%\caption{}
%\label{subfig:ddg}
\end{subfigure}
\caption{Example of a library/client program and the corresponding program-valid data-dependence graph.
The library consists of method $g()$ which has a callback function $f(x,y)$.
The client implements $f(x,y)$ either as $f_1(x,y)$ or $f_2(x,y)$.
The parenthesis-labeled edge model context-sensitive dependencies on parameter passing and return.
Depending on the implementation of $f$, there is a data dependence of the variable $p$ on $y$.
}
\label{fig:lib_example}
\end{figure}
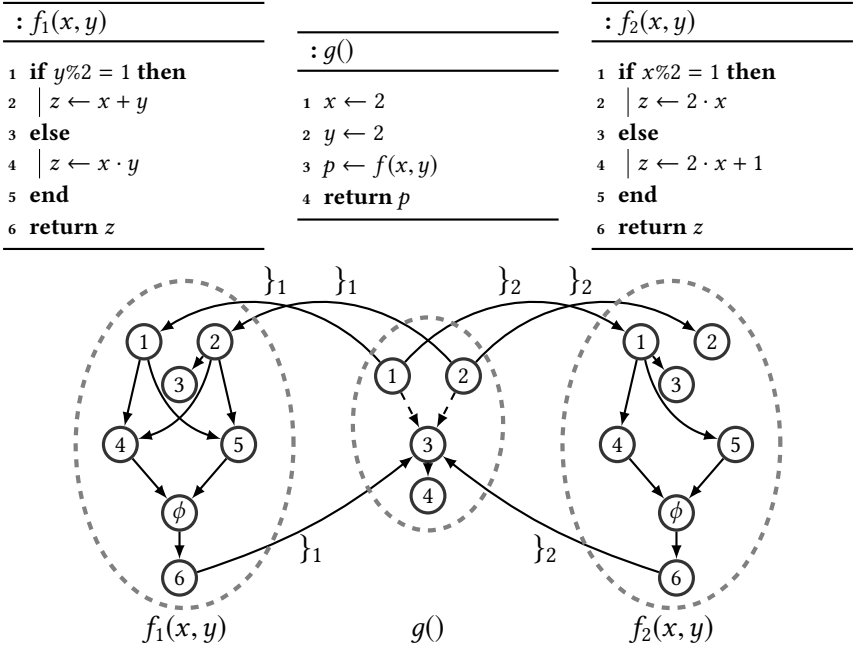

\subsection{Library/Client Dyck Reachability on Program-valid Graphs}\label{subsec:library}

We are now ready to present our method for computing library summaries on program-valid graphs in order to speed up the client-side Dyck reachability.
The approach is very similar to the work of \cite{CIP15} for data-flow analysis of recursive state machines.

\noindent{\bf Outline of our approach.}
Our approach consists of the following conceptual steps.
We let the input graph $G=(V,E)$ be any program-valid graph of constant treewidth,
with a partitioning of $V$ into the library component $V^1$ and the client component $V^2$.
Since $G$ is program-valid, it has a constant-treewidth, program-valid partitioning $\Partition$, and we consider $\Partition^{1}$ to be the restriction of $\Partition$ to the set $V^1$.
Hence we have  $\Partition^{1}=\{V_1, \dots V_{\ell}\}$ be a program-valid partitioning of $G\restr{V^1}$, which also has constant treewidth. Our approach consists of the following steps.

\begin{compactenum}
\item We construct a local graph $G_i=(V_i, E_i)$ and the corresponding maximal local graph $G'_i=(V_i, E'_i)$ for each $V_i\in \Partition$.
Recall that $G'_i$ is a conventional graph, since, by definition, $E'_i$ contains only $\epsilon$-labeled edges.
Since $\Partition$ has constant treewidth, each graph $G'_i$ has constant treewidth, and we construct a tree decomposition $\Tree(G'_i)$.
\item We exploit the constant-treewidth property of each $G'_i$ to build a data structure $\DataStructure$ which supports the following two operations:
(i) Querying whether a node $v$ is reachable from a node $u$ in $G'_i$, and
(ii) Updating $G_i$ by inserting a new edge $(x,y)$.
Moreover, each such operation is fast, i.e., it is performed in $O(\log n_i)$ time.
\item Recall that $V^1$, $V^2$ are the library and client partitions of $G$, respectively.
In the preprocessing phase, we use the data structure $\DataStructure$ to preprocess $G\restr{V^1}$ so that any pair of library nodes that is Dyck-reachable in $G\restr{V^1}$ is discovered and can be queried fast. Hence this library-side reachability information serves as the summary on the library side.
\item In the query phase, we use $\DataStructure$ to process the whole graph $G$, using the summaries computed in the preprocessing phase.
\end{compactenum}

\noindent{\bf Step 1. Construction of the local graphs $G_i$ and the tree decompositions.}
The local graphs $G_i$ are extracted from $G\restr{V^1}$ by means of its program-valid partitioning $\Partition^1=\{V_1,\dots V_{\ell}\}$.
We consider this partitioning as part of the input, since every local graph $G_i$ in reality corresponds to a unique method of the input program represented by $G$.
Let $n_i=|V_i|$.
The maximal local graphs $G'_i=(V_i,E'_i)$ are constructed as defined in Section~\ref{subsec:tree_decompositions}.
Each tree decomposition $\Tree(G'_i)$ is constructed in $O(n_i)$ time using Theorem~\ref{them:tree_dec}.
Observe that since $E_i\subseteq E'_i$ (i.e., $G_i$ is a subgraph of its maximal counterpart $G'_i$), 
$\Tree(G'_i)$ is also a tree decomposition of $G_i$.
We define $\Tree(G_i)=\Tree(G'_i)$ for all $1\leq i \leq  \ell$.

\noindent{\bf Step 2. Description of the data structure $\DataStructure$.}
Here we describe the data structure $\DataStructure$, which is built for a conventional graph $G_i=(V_i,E_i)$ (i.e., $E_i$ has only $\epsilon$-labeled edges) and its tree decomposition $\Tree(G_i)$. 
The purpose of $\DataStructure$ is to handle reachability queries on $G_i$.
The data structure supports three operations, given in Algorithm~\ref{algo:preprocess}, Algorithm~\ref{algo:update} and Algorithm~\ref{algo:query}.
\begin{compactenum}
\item The $\DataStructure.\Preprocessalgo$ (Algorithm~\ref{algo:preprocess}) operation builds the data structure for $G_i$.
\item The $\DataStructure.\Updatealgo$ (Algorithm~\ref{algo:update}) updates the graph $G_i$ with a new edge $(x,y)$, provided that there exists a bag $\Bag$ such that $x,y\in \Bag$.
\item The $\DataStructure.\Queryalgo$ (Algorithm~\ref{algo:query}) takes as input a pair of nodes $x,y$ and returns $\True$ iff $y$ is reachable from $x$ in $G_i$, considering all the update operations performed so far.
\end{compactenum}

\begin{minipage}[t]{.47\textwidth}
\begin{algorithm}[H]
\small
\SetInd{0.4em}{0.4em}
\DontPrintSemicolon
%\setstretch{1.05}
\caption{$\DataStructure.\Preprocessalgo$}\label{algo:preprocess}
\KwIn{A tree-decomposition $\Tree(G_i)$}
\BlankLine
Traverse $\Tree(G_i)$ bottom up\\
\ForEach{encountered bag $\Bag$}{
Construct the graph $G(\Bag)=(\Bag, \ReachMap(\Bag))$\\
Compute the transitive closure $G^*(\Bag)$\\
\ForEach{$(u,v)\in \Bag$}{
\uIf{$u\Path v$ in $G^*(\Bag)$}{
Insert $u,v$ in $\ReachMap$
}
}
}
\end{algorithm}
\end{minipage}
\qquad
\begin{minipage}[t]{.47\textwidth}
\begin{algorithm}[H]
\small
\SetInd{0.4em}{0.4em}
\DontPrintSemicolon
%\setstretch{1.05}
\caption{$\DataStructure.\Updatealgo$}\label{algo:update}
\KwIn{A new edge $(x,y)$}
\BlankLine
Traverse $\Tree(G)$ from $\Bag_{(u,v)}$ to the root\\
\ForEach{encountered bag $\Bag$}{
Construct the graph $G(\Bag)=(\Bag, \ReachMap(\Bag))$\\
Compute the transitive closure $G^*(\Bag)$\\
\ForEach{$u,v\in \Bag$}{
\uIf{$u\Path v$ in $G^*(\Bag)$}{
Insert $(u,v)$ in $\ReachMap$
}
}
}
\end{algorithm}
\end{minipage}

\begin{minipage}[t]{.44\textwidth}
\begin{algorithm}[H]
\small
\SetInd{0.4em}{0.4em}
\DontPrintSemicolon
%\setstretch{1.05}
\caption{$\DataStructure.\Queryalgo$}\label{algo:query}
\KwIn{A pair of nodes $x,y$}
\BlankLine
Let $X\gets \{x\}, Y\gets \{y\}$\\
Traverse $\Tree(G)$ from $\Bag_x$ to the root\\
\ForEach{encountered bag $\Bag$}{
\ForEach{$u,v\in \Bag$}{
\uIf{$u\in X$ and $(u,v)\in \ReachMap$}{\label{line:query_add}
Add $v$ to $X$
}
}
}
Traverse $\Tree(G)$ from $\Bag_y$ to the root\\
\ForEach{encountered bag $\Bag$}{
\ForEach{$u,v\in \Bag$}{
\uIf{$v\in Y$ and $(u,v)\in \ReachMap$}{
Add $u$ to $Y$
}
}
}
\Return $\True$ iff $X\cap Y \neq\emptyset$\label{line:query_intersection}
\end{algorithm}
\end{minipage}
\qquad
\begin{minipage}[t]{.51\textwidth}
\begin{algorithm}[H]
\small
\SetInd{0.4em}{0.4em}
\DontPrintSemicolon
%\setstretch{1.05}
\caption{$\RSMalgo$}\label{algo:process}
\KwIn{Method graphs $(G_i=(V_i, E_i))_{1\leq i\leq \ell}$}
\BlankLine
\ForEach{$1\leq i\leq \ell$}{
Construct $\Tree(G_j)$\label{line:construct_td}\\
Run $\DataStructure.\Preprocessalgo$ on $\Tree(G_i)$\label{line:preprocess_all}
}
$\Pool\gets \{G_1,\dots G_{\ell}\}$\\
\While{$\Pool\neq \emptyset$}{\label{line:process_main_loop}
Extract $G_j$ from $\Pool$\label{line:extract_pool}\\
\ForEach{$u\in V_j \cap V_e, v\in V_j \cap V_x$}{\label{line:for_entry_exit}
\uIf{$\DataStructure.\Queryalgo(u,v)$}{\label{line:if_entry_exit}
\ForEach{$x, y: (x,u,\OpenParenthesis_i), (v,y,\OpenParenthesis_i)\in E$}{\label{line:for_call_return}
Let $G_r=(V_r, E_r)$ be the graph s.t. $x,y\in V_r$\label{line:let_gr}\\
\uIf{not $\DataStructure.\Queryalgo(x,y)$}{\label{line:if_new_query}
Run $\DataStructure.\Updatealgo$ on $\Tree(G_r)$ on $(x,y)$\label{line:update}\\
Insert $G_r$ in $\Pool$\label{line:insert_in_pool}
}
}
}
}
}
\end{algorithm}
\end{minipage}

\noindent{\bf The reachability set $\ReachMap$.}
The data structure $\DataStructure$ is built by storing a reachability set $\ReachMap$ between pairs of nodes. 
The set $\ReachMap$ has the crucial property that it stores information only between pairs of nodes that appear in some bag of $\Tree(G_i)$ together.
That is, $\ReachMap=\subseteq  \bigcup_{\Bag} \Bag\times\Bag$. Given a bag $\Bag$, we denote by $\ReachMap(\Bag)$ the restriction of $\ReachMap$ to the nodes of $\Bag$.
The reachability set is stored as a collection of $2\sum_i\cdot n_i$ sets $\ReachMap^{F}(u)$ and $\ReachMap^{B}(u)$, one for every node $u\in V_i$.
In turn, the set $\ReachMap^{F}(u)$ (resp. $\ReachMap^{B}(u)$) will store the nodes in $\Bag_u$ (recall that $\Bag_u$ is the root bag of node $u$) for which it has been discovered that can be reached from $u$ (resp., that can reach $u$).
It follows directly from the definition of tree decompositions that if $(u,v)\in E_i$ is an edge of $G_i$ then  $u\in \Bag_v$ or $v\in \Bag_u$.
Given a bag $\Bag$ and nodes $u,v\in \Bag$, querying whether $(u,v)\in \ReachMap$ reduces to testing whether $v\in \ReachMap^{F}(u)$ or $u\in \ReachMap^{B}(v)$. Similarly, inserting $(u,v)$ to $\ReachMap$ reduces to inserting either $v$ to $\ReachMap^F(u)$ (if $v\in \Bag_u$),
or $u$ to $\ReachMap^{B}(v)$ (if $u\in \Bag_v$).

\begin{remark}\label{rem:reachability_map}
The map $\ReachMap$ requires $O(n)$ space. Since each $G_i$ is a constant-treewidth graph, every insert and query operation on $\ReachMap$ requires $O(1)$ time.
\end{remark}

\noindent{\bf Correctness and complexity of $\DataStructure$.}
Here we establish the correctness and complexity of each operation of $\DataStructure$.

It is rather straightforward to see that for every pair of nodes $(u,v)\in \ReachMap$, we have that $v$ is reachable from $u$.
The following lemma states a kind of weak completeness: if $v$ is reachable from $u$ via a path of specific type,
then $(u,v)\in \ReachMap$.
Although this is different from strong completeness, which would require that $(u,v)\in \ReachMap$ whenever $v$ is reachable from $u$,
it is sufficient for ensuring completeness of the $\DataStructure.\Queryalgo$ algorithm.

\noindent{\bf Left-right-contained paths.}
We introduce the notion of left-right contained paths, which is crucial for stating the correctness of the data structure $\DataStructure$.
Given a bag $\Bag$ of $\Tree(G_i)$, we say that a path $P:x\Path y$ is 
\emph{left-contained} in $\Bag$ if for every node $w\in P$, if $w\neq x$, we have that $\Bag_w\in T(\Bag)$.
Similarly, $P$ is \emph{right-contained} in $\Bag$ if for every node $w\in P$, if $w\neq y$, we have that $\Bag_w\in T(\Bag)$.
Finally, $P$ is \emph{left-right-contained} in $\Bag$ if it is both left-contained and right-contained in $\Bag$.

\begin{lemma}\label{lem:ds_correctness}
The data structure $\DataStructure$ maintains the following invariant.
For every bag $\Bag$ and pair of nodes $u,v\in \Bag$,
if there is a $P_u^v:u\Path v$ which is left-right contained in $\Bag$,
then after $\DataStructure.\Preprocessalgo$ has processed $\Bag$,
we have $(u,v)\in \ReachMap$.
\end{lemma}

It is rather straightforward that at the end of $\DataStructure.\Queryalgo$, for every node $w\in X$ 
(resp. $w\in Y$) we have that $w$ is reachable from $x$ (resp. $y$ is reachable from $w$).
This guarantees that if $\DataStructure.\Queryalgo$ returns $\True$, then $y$ is indeed reachable from $x$,
via some node $w\in X\cap Y$ (recall that the intersection is not empty, due to Line~\ref{line:query_intersection}).
The following two lemmas state completeness, namely that if $y$ is reachable from $x$, then $\DataStructure.\Queryalgo$ will return $\True$,
and the complexity of $\DataStructure$ operations.

\begin{lemma}\label{lem:query_correctness}
On input $x,y$, if $y$ is reachable from $x$, then $\DataStructure.\Queryalgo$ returns $\True$.
\end{lemma}

\begin{lemma}\label{lem:ds_complexity}
$\DataStructure.\Preprocessalgo$ requires $O(n_i)$ time.
Every call to $\DataStructure.\Updatealgo$ and $\DataStructure.\Queryalgo$ requires $O(\log n_i)$ time.
\end{lemma}

\noindent{\bf Step 3. Preprocessing the library graph $G\restr{V^1}$.}
Given the library subgraph $G\restr{V^1}$ and one copy of the data structure $\DataStructure$ for each local graph $G_i$ of $G\restr{V^1}$, 
the preprocessing of the library graph
is achieved via the algorithm $\RSMalgo$, which is presented in Algorithm~\ref{algo:process}.
In high level, $\RSMalgo$ initially builds the data structure $\DataStructure$ for each local graph $G_i$ using $\DataStructure.\Preprocessalgo$.
Afterwards, it iteratively uses $\DataStructure.\Queryalgo$ to test whether there exists a local graph $G_j$ and two nodes $u\in V_j\cap V_e$, $v\in V_j\cap V_x$ such that $v$ is reachable from $u$ in $G_j$.
If so, the algorithm iterates over all nodes $x,y$ such that $(x,u,\OpenParenthesis_i)\in E$ and $(v,y,\CloseParenthesis_i)\in E$, and uses a $\DataStructure.\Queryalgo$ operation to test whether $y$ is reachable from $x$ in their respective local graph $G_r$.
If not, then $\RSMalgo$ uses a $\DataStructure.\Updatealgo$ operation to insert the edge $x,y$ in $G_r$.
Since this new edge might affect the reachability relations among other nodes in $V_r$, the graph $G_r$ is inserted in $\Pool$ for further processing.
See Algorithm~\ref{algo:process} for a formal description.
The following two lemmas state the correctness and complexity of $\RSMalgo$. 

\begin{lemma}\label{lem:process_correctness}
At the end of $\RSMalgo$, for every graph $G_i=(V_i, E_i)$ and pair of nodes $u,v\in V_i$,
we have that $v$ is reachable from $u$ in $G\restr{V^1}$ iff $\DataStructure.\Queryalgo$ returns $\True$.
\end{lemma}

\begin{lemma}\label{lem:process_complexity}
Let $n=\sum_i n_i$, and $k_1$ be the number of labels appearing in $E\subseteq V^1\times V^1\times \Alphabet_k$ (i.e., $k_1$ is the number of call sites in $G\restr{V^1}$).
$\RSMalgo$ requires $O(n+k_1\cdot \log n)$ time.
\end{lemma}

\noindent{\bf Step 4. Library/Client analysis.}
We are ready to describe the library summarization for Library/Client Dyck reachability.
Let $G=(V,E)$ be the program-valid graph representing library and client code, and $V^1, V^2$ a partitioning of $V$
to library and client nodes. 
\begin{compactenum}
\item In the preprocessing phase, the algorithm $\RSMalgo$ is used to preprocess $G\restr{V^1}$.
Note that since $G$ is a program valid graph, so is $G\restr{V^1}$, hence $\RSMalgo$ can execute on $G\restr{V^1}$.
The summaries created are in form of $\DataStructure.\Updatealgo$ operations performed on edges $(x,y)$.
\item In the querying phase, the set $V^2$ is revealed, and thus the whole of $G$.
Hence now $\RSMalgo$ processes $G$, without using $\DataStructure.\Preprocessalgo$ on the graphs $G_i$ that correspond to library methods,
as they have already been processed in step~1. Note that the graphs $G_i$ that correspond to library methods are used for querying and updating.
\end{compactenum}
It follows immediately from Lemma~\ref{lem:process_correctness} that at the end of the second step, 
for every local graph $G_i=(V_i, E_i)$ of the client graph, for every pair of nodes $u,v \in V_i$,
$v$ is Dyck-reachable from $u$ in the program-valid graph $G$ if and only if $\DataStructure.\Queryalgo$ returns $\True$ on input $u,v$.

Now we turn our attention to complexity.
Let $n_1=|V^1|$ and $n_2=|V^2|$.
By Lemma~\ref{lem:process_complexity}, the time spent for the first step is, $O(n_1+ k_1\cdot \log n_1)$,
and the time spent for the second step is $O(n_2 + k_1\cdot \log n_1 + k_2\cdot \log n_2)$.

\noindent{\bf Constant-time queries.}
Recall that our task is to support $O(1)$-time queries about the Dyck reachability of pairs of nodes on the client subgraph $G\restr{V^2}$.
As Lemma~\ref{lem:process_correctness} shows, after $\RSMalgo$ has finished, each such query costs $O(\log n_2)$ time.
We use existing results for reachability queries on constant-treewidth graphs~\cite[Theorem~6]{CIP16} which allow us to reduce the query time to $O(1)$,
while spending $O(n_2)$ time in total to process all the graphs.

\begin{theorem}\label{them:library}
Consider a $\Alphabet_k$-labeled program-valid graph $G=(V,E)$ of constant program-valid treewidth,
and the library and client subgraphs $G^1=(V^1, E^1)$ and $G^2=(V^2, E^2)$.
For $i\in \{1,2\}$ let $n_i=|V^i|$ be the number of nodes, and $k_i$ be the number of call sites in each graph $G^i$, with $k_1+k_2=k$.
The algorithm $\DynamicDyck$ requires 
\begin{compactenum}
\item $O(n_1+k_1\cdot \log n_1)$ time and $O(n_1)$ space in the preprocessing phase, and
\item $O(n_2 + k_1\cdot \log n_1 + k_2\cdot \log n_2)$ time and $O(n_1+n_2)$ space in the query phase,
\end{compactenum}
after which pair reachability queries are handled in $O(1)$ time.

\end{theorem}

\section{Experimental Results}\label{sec:experiments}

In this section we report on experimental results obtained for the problems 
of (i) alias analysis via points-to analysis on SPGs, and (ii) library/client data-dependence analysis.

\subsection{Alias Analysis}\label{subsec:experiments_bidirected}

\noindent{\bf Implementation.}
We have implemented our algorithm $\BidirectedAlgo$ in C++ and evaluated its performance in 
performing Dyck reachability on bidirected graphs.
The algorithm is implemented as presented in Section~\ref{sec:bidirected},
together with the preprocessing step that handles the $\epsilon$-labeled edges.
Besides common coding practices we have performed no engineering optimizations.
We have also implemented~\cite[Algorithm~2]{Zhang13}, including the Fast-Doubly-Linked-List (FDLL),
which was previously shown to be very efficient in practice.

\noindent{\bf Experimental setup.}
In our experimental setup we used the DaCapo-2006-10-MR2 suit~\cite{Blackburn06}, which contains 11 real-world benchmarks.
We used the tool reported in~\cite{Yan11} to extract the Symbolic Points-to Graphs (SPGs), which in turn uses Soot~\cite{Soot} to process input Java programs.
Our approach is similar to the one reported in~\cite{PLDI36,Yan11,Zhang13}.
The outputs of the two compared methods were verified to ensure validity of the results.
No compiler optimizations were used.
All experiments were run on a Windows-based laptop with an Intel Core i7-5500U \@ $2.40$ GHz CPU and $16$ GB of memory, without any compiler optimizations.

\noindent{\bf SPGs and points-to analysis.}
For the sake of completeness, we outline the construction of SPGs and the reachability relation they define.
A more detailed exposition can be found in~\cite{PLDI36,Yan11,Zhang13}.
An SPG is a graph, the node set of which consists of the following three subsets:
(i)~variable nodes $\mathcal{V}$ that represent variables in the program,
(ii)~allocation nodes $\mathcal{O}$ that represent objects constructed with the \emph{new} expression, and
(iii)~symbolic nodes $\mathcal{S}$ that represent abstract heap objects.
Similarly, there are three types of edges, as follows, where $\mathsf{Fields}=\{f_i\}_{1\leq i\leq k}$ denotes the set of all fields of composite data types.

\begin{compactenum}
\item Edges of the form $\mathcal{V} \times\mathcal{O}\times\{\epsilon\}$ represent the objects that variables point to.
\item Edges of the form $\mathcal{V}\times \mathcal{S} \times \{\epsilon\} $ represent the abstract heap objects that variables point to.
\item Edges of the form $\left(\mathcal{O} \cup \mathcal{S}\right) \times \left(\mathcal{O} \cup \mathcal{S}\right) \times \mathsf{Fields}$ represent the fields of objects that other objects point to.
\end{compactenum}

We note that since we focus on context-insensitive points-to analysis, we have not included edges that model calling context in the definition of the SPG.
Additionally, only the forward edges labeled with $f_i$ are defined explicitly, and the backwards edges labeled with $\ov{f}_i$ are implicit, since the SPG is treated as bidirected.
Memory aliasing between two objects $o_1,o_2\in \mathcal{S}\cup \mathcal{O}$ occurs when there is a path $o_1\Path o_2$, such that every opening field access $f_i$ is properly matched by a closing field access $\ov{f}_i$.
Hence the Dyck grammar is given by $\StartNonTerminal\to \StartNonTerminal~ \StartNonTerminal ~ | ~ f_i ~ \StartNonTerminal~\ov{f}_i ~ |  ~ \epsilon$.
This allows to infer the objects that variable nodes can point to via composite paths that go through many field assignments.
See Figure~\ref{fig:spg} for a minimal example.
\begin{figure}[!h]
\centering
\begin{tikzpicture}[thick, >=latex,
pre/.style={<-,shorten >= 1pt, shorten <=1pt, thick},
post/.style={->,shorten >= 1pt, shorten <=1pt,  thick},
und/.style={very thick, draw=gray},
node1/.style={circle, minimum size=3.5mm, draw=black!100, line width=1pt, inner sep=0},
node2/.style={rectangle, minimum size=3.5mm, draw=black!100, fill=white!100, very thick, inner sep=0},
virt/.style={circle,draw=black!50,fill=black!20, opacity=0}]

\newcommand{\xdisposition}{0}
\newcommand{\ydisposition}{-0.2}
\newcommand{\xstep}{0.45}
\newcommand{\ystep}{0.35}

\node[] at (\xdisposition + 0*\xstep, \ydisposition + 0*\ystep) {$z.f=x$};
\node[] at (\xdisposition + 0*\xstep, \ydisposition + -1*\ystep) {$y=z.f$};

\renewcommand{\xdisposition}{3}
\renewcommand{\ydisposition}{0}
\renewcommand{\xstep}{1}
\renewcommand{\ystep}{0.8}

\node[node1] (x) at (\xdisposition + 0*\xstep, \ydisposition + 0*\ystep) {$x$};
\node[node1] (z) at (\xdisposition + 1*\xstep, \ydisposition + 0*\ystep) {$z$};
\node[node1] (y) at (\xdisposition + 2*\xstep, \ydisposition + 0*\ystep) {$y$};

\node[node2] (xx) at (\xdisposition + 0*\xstep, \ydisposition + -1*\ystep) {$x$};
\node[node2] (zz) at(\xdisposition + 1*\xstep, \ydisposition + -1*\ystep) {$z$};
\node[node2] (yy) at(\xdisposition + 2*\xstep, \ydisposition + -1*\ystep) {$y$};

\draw[->, very thick] (x) to (xx);
\draw[->, very thick] (y) to (yy);
\draw[->, very thick] (z) to (zz);

\draw[->, very thick] (zz) to node[above]{$f$} (xx);
\draw[->, very thick] (zz) to node[above]{$f$} (yy);

\end{tikzpicture}
\caption{A minimal program and its (bidirected) SPG. Circles and squares represent variable nodes and object nodes, respectively. Only forward edges are shown.}\label{fig:spg}
\end{figure}
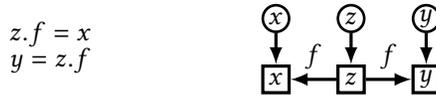

\noindent{\bf Analysis of results.}
The running times of the compared algorithms are shown in Table~\ref{tab:bidirected_experiments}
We can see that the algorithm proposed in this work is much faster than the existing algorithm of~\cite{Zhang13}
in all benchmarks.
The highest speedup is achieved in benchmark \emph{luindex}, where our algorithm is 13x times faster.
We also see that all times are overall small. 
%This reduction in running time is welcome trade-off at the cost of making the points-to analysis context insensitive,
%and thus less precise.

\begin{table}
\centering
\renewcommand{\arraystretch}{0.9}
\small
\caption{Comparison between our algorithm and the existing from ~\cite{Zhang13}.
The first three columns contain the number of fields (Dyck parenthesis),
nodes and edges in the SPG of each benchmark.
The last two columns contain the running times, in seconds.}
\label{tab:bidirected_experiments}
\begin{tabular}{|c|||c|c|c||c|c|}
\hline
\textbf{Benchmark} &  \textbf{Fields} & \textbf{Nodes} & \textbf{Edges} &  \textbf{Our Algorithm} & \textbf{Existing Algorithm}\\
\hline
\hline
antlr & 172 & 13708 & 23547  & 0.428783 & 1.34152\\
\hline
bloat & 316 & 43671 & 103361  & 17.7888 & 34.6012\\
\hline
chart & 711 & 53500 & 91869  & 8.99378 & 34.9101\\
\hline
eclipse & 439 & 34594 & 52011  & 3.62835 & 12.7697\\
\hline
fop & 1064 & 101507 & 178338  & 42.5447 & 148.034\\
\hline
hsqldb & 43 & 3048 & 4134  & 0.012899 & 0.073863\\
\hline
jython & 338 & 56336 & 167040  & 40.239 & 55.3311\\
\hline
luindex & 167 & 9931 & 14671  & 0.068013 & 0.636346\\
\hline
lusearch & 200 & 12837 & 21010  & 0.163561 & 1.12788\\
\hline
pmd & 357 & 31648 & 58025  & 2.21662 & 8.92306\\
\hline
xalan & 41 & 2342 & 2979  & 0.006626 & 0.045144\\
\hline
\end{tabular}
\end{table}

%\begin{figure}
%\centerline{
%\includegraphics[scale=0.3]{figures/no_libs_dump_eps_merging.pdf}
%}
%\caption{Running time  of our algorithm vs~\cite{Zhang13} for context-insensitive field-sensitive points-to analysis on SPGs of various benchmarks.
%The  top row shows the total time (in ms) taken for the slowest method to perform the analysis.
%The total time is taken as the sum of the time spent in analyzing library and client code.
%The y-axis shows the percentage of time that each method took as compared to the slowest method.
%}
%\label{fig:bidirected_time}
%\end{figure}

\subsection{Library/Client Data-dependence Analysis}\label{subsec:experiments_library}

\noindent{\bf Implementation.}
We have implemented our algorithm $\DynamicDyck$ in Java and evaluated its performance in
performing Library/Client data-dependency analysis via Dyck reachability.
Our algorithm is built on top of Wala~\cite{Wala}, and is implemented as presented in Section~\ref{sec:library}.
Besides common coding practices we have performed no engineering optimizations.
We used the LibTW library~\cite{Dijk06} for computing the tree decompositions of the input graphs, under the \emph{greedy degree} heuristic.

\noindent{\bf Experimental setup.}
We have used the tool of \cite{Tang15} for obtaining the data-dependence graphs of Java programs.
In turn, that tool uses Wala~\cite{Wala} to build the graphs, and specifies the parts of the graph that correspond to library and client code.
Java programs are suitable for Library/Code analysis, since the ubiquitous presence of callback functions
makes the library and client code interdependent, so that the two sides cannot be analyzed in isolation.
Our algorithm was compared with the TAL reachability and CFL reachability approach,
as already implemented in~\cite{Tang15}. 
The comparison was performed in terms of running time and memory usage, first for the analysis of library code to produce the summaries,
and then for the analysis of the library summaries with the client code.
The outputs of all three methods were compared to ensure validity of the results.
The measurements for our algorithm include the time and memory used for computing the tree decompositions.
All experiments were run on a Windows-based laptop with an Intel Core i7-5500U \@ $2.40$ GHz CPU and $16$ GB of memory, without any compiler optimizations.

\noindent{\bf Benchmarks.}
Our benchmark suit is similar to that of~\cite{Tang15}, consisting of $12$ Java programs from SPECjvm2008~\cite{SPECjvm2008},
together with $4$ randomly chosen programs from GitHub~\cite{Github}.
We note that as reported in~\cite{Tang15}, they are unable to handle the benchmark \emph{serial} from SPECjvm2008, 
due to out-of-memory issues when preprocessing the library (recall that the space bound for TAL reachability is $O(n^4)$).
In contrast, our algorithm handles \emph{serial} easily, and is thus included in the experiments.

\noindent{\bf Analysis of results.}
Our experimental comparison is depicted in Table~\ref{tab:libraries_time} for running time and Table~\ref{tab:libraries_space} for memory usage
We briefly discuss our findings.

\emph{Treewidth.}
First, we comment on the treewidth of the obtained data-dependence graphs, which is reported on Table~\ref{tab:libraries_time} and Table~\ref{tab:libraries_space}.
Recall that our interest is not on the treewidth of the whole data-dependence graph, but on the treewidth of its program-valid partitioning,
which yields a subgraph for each method of the input program.
In each line of the tables we report the \emph{maximum} treewidth of each benchmark, i.e. the maximum treewidth over the subgraphs of its program-valid partitioning.
We see that the treewidth is typically very small (i.e., in most cases it is $5$ or $6$) in both library and client code.
One exception is the client of mpegaudio, which has large treewidth. Observe that even this corner case of large treewidth was easily handled by our algorithm.

\emph{Time.}
% Figure~\ref{fig:cfl_total_times} compares the total running times for analyzing library and client code.
Table~\ref{tab:libraries_time} shows the time spent by each algorithm for analyzing library and client code separately.
We first focus on total time, taken as the sum of the times spent by each algorithm in the library and client graph of each benchmark.
We see that in every benchmark, our algorithm significantly outperforms both TAL and CFL reachability, reaching a 10x-speedup compared to TAL (in \emph{mpegaudio}),
and 5x-speedup compared to CFL reachability (in \emph{helloworld}).
Note that the benchmark \emph{serial} is missing from the figure, as TAL reachability runs out of memory.
The benchmark is found on Table~\ref{tab:libraries_time}, where our algorithm achieves a 630x-speedup compared to CFL reachability.

%Table~\ref{tab:libraries_time} shows the time spent by each algorithm for analyzing library code and client code separately.
We now turn our attention to the trade-off between library preprocessing and client querying times.
Here, the advantage of TAL over CFL reachability is present for handling client code.
However, even for client code our algorithm is faster than TAL in all cases except one, and reaches even a 30x-speedup over TAL (in \emph{sunflow}).
Finally, observe that in all cases, the total running time of our algorithm on library and client code combined is much smaller than each of the other methods on library code alone.

\noindent\emph{Memory.} 
Table~\ref{tab:libraries_space} compares the total memory used for analyzing library and client code.
We see that our algorithm significantly outperforms both TAL and CFL reachability in all benchmarks.
%Table~\ref{tab:libraries_space} shows the memory used by each algorithm for analyzing library and client code separately.
Again, TAL uses more memory that CFL in the preprocessing of libraries, but less memory when analyzing client code.
However, our algorithm uses even less memory than TAL in all benchmarks.
The best performance gain is achieved in \emph{serial}, where TAL runs out of memory after having consumed more than $12$~GB.
For the same benchmark, CFL reachability uses more than $4.3$~GB.
In contrast, our algorithm uses only $130$~MB, thus achieving a 33x-improvement over CFL, and at least a 90x-improvement over TAL.
We stress that for memory usage, these are tremendous gains.
Finally, observe that for each benchmark, the maximum memory used by our algorithm for analyzing library and client code is smaller than the minimum memory used between library and client, by each of the other two methods.

\noindent{\bf Improvement independent of callbacks.}
We note that in contrast to TAL reachability, the improvements of our algorithm are not restricted to the presence of callbacks.
Indeed, the algorithms introduced here significantly outperform the CFL approach even in the presence of no callbacks.
This is evident from Table~\ref{tab:libraries_time}, which shows that our algorithm processes the library graphs much faster than both CFL and TAL reachability.

\begin{table}
\centering
\small
\renewcommand{\arraystretch}{0.9}
\setlength\tabcolsep{3.5pt}
\caption{Running time  of our algorithm vs the TAL and CFL approach for data-dependence analysis with library summarization.
Times are in milliseconds. MEM-OUT indicates that the algorithm run out of memory.
The number of nodes and treewidth reflects the average and maximum case, respectively, among all methods in each benchmark.}\label{tab:libraries_time}
\begin{tabular}{|c|||c|c||c|c|||c|c||c|c||c|c|}
\hline
 & \multicolumn{2}{c||}{Nodes} & \multicolumn{2}{c|||}{TW} & \multicolumn{2}{|c||}{\textbf{Our Algorithm}} & \multicolumn{2}{|c||}{\textbf{TAL}} & \multicolumn{2}{|c|}{\textbf{CFL}} \\
 \hline
\textbf{Benchmark} & Lib. & Cl. & Lib. & Cl. & Lib. & Cl. & Lib. & Cl. & Li. & Cl. \\
\hline
\hline
helloworld & 16003 & 296 & 5 &3 & 229 & 5 & 1044 & 31 & 855 & 578\\
\hline
check & 16604 & 3347 & 5 & 4 & 228 & 54 & 1062 & 72 & 821 & 620\\
\hline
compiler & 16190 & 536 &5 & 3 & 248 & 11 & 995 & 57 & 876 & 572\\
\hline
sample & 3941 & 28 & 4 & 1 & 86 & 1 & 258 & 14 & 368 & 113\\
\hline
crypto & 20094 & 3216 & 5 & 5 & 273 & 66 & 1451 & 196 & 961 & 776\\
\hline
derby & 23407 & 1106 & 6 & 3 & 389 & 22 & 1301 & 83 & 1003 & 1100\\
\hline
mpegaudio & 28917 & 27576 &5 & 24 & 204 & 177 & 5358 & 253 & 1864 & 1586\\
\hline
xml & 71474 & 2312 & 5 & 3 & 489 & 115 & 5492 & 100 & 1891 & 2570\\
\hline
mushroom & 3858 & 7 & 4 & 1 & 86 & 1 & 230 & 14 & 349 & 124\\
\hline
btree & 6710 & 1103 &4 & 4 & 144 & 34 & 583 & 111 & 571 & 197\\
\hline
startup & 19312 & 621 &5 & 3 & 279 & 17 & 1651 & 110 & 1087 & 946\\
\hline
sunflow & 15615 & 85 & 5 & 2 & 217 & 1 & 1073 & 31 & 811 & 549\\
\hline
compress & 16157 & 1483 & 5 & 3 & 240 & 23 & 1119 & 112 & 783 & 999\\
\hline
parser & 7856 & 112 & 4 & 1 & 172 & 3 & 443 & 21 & 572 & 241\\
\hline
scimark & 16270 & 2027 &5 & 5 & 220 & 34 & 1004 & 70 & 805 & 595\\
\hline
serial & 69999 & 468 & 8 & 3 & 440 & 9 & \footnotesize{MEM-OUT} & \footnotesize{MEM-OUT} & 117147 & 165958\\
%\hline
%\hline
%\textbf{Total} & &  & & &  \textbf{3504} & \textbf{564} & \textbf{23064} & \textbf{1275} & \textbf{13617} & \textbf{11566}\\
\hline
\end{tabular}
\end{table}

\begin{table}
\centering
\small
\renewcommand{\arraystretch}{0.9}
\setlength\tabcolsep{4pt}
\caption{Memory usage of our algorithm vs the TAL and CFL approach for data-dependence analysis with library summarization.
Memory usage is in Megabytes. MEM-OUT indicates that the algorithm run out of memory.
The number of nodes and treewidth reflects the average and maximum case, respectively, among all methods in each benchmark.}\label{tab:libraries_space}
\begin{tabular}{|c|||c|c||c|c|||c|c||c|c||c|c|}
\hline
 & \multicolumn{2}{c||}{Nodes} & \multicolumn{2}{c|||}{TW} & \multicolumn{2}{|c||}{\textbf{Our Algorithm}} & \multicolumn{2}{|c||}{\textbf{TAL}} & \multicolumn{2}{|c|}{\textbf{CFL}} \\
 \hline
\textbf{Benchmark} & Lib. & Cl. & Lib. & Cl. & Lib. & Cl. & Lib. & Cl. & Li. & Cl. \\
\hline
\hline
helloworld & 16003 & 296 & 5 &3 & 31 & 27 & 321 & 44 & 104 & 126\\
\hline
check & 16604 & 3347 & 5 & 4 & 34 & 31 & 336 & 89 & 132 & 184\\
\hline
compiler & 16190 & 536 &5 & 3 & 31 & 28 & 329 & 44 & 108 & 137\\
\hline
sample & 3941 & 28 & 4 & 1 & 19 & 16 & 232 & 59 & 59 & 64\\
\hline
crypto & 20094 & 3216 & 5 & 5 & 45 & 45 & 261 & 61 & 127 & 188\\
\hline
derby & 23407 & 1106 & 6 & 3 & 46 & 41 & 600 & 88 & 204 & 265\\
\hline
mpegaudio & 28917 & 27576 &5 & 24 & 96 & 96 & 516 & 219 & 262 & 397\\
\hline
xml & 71474 & 2312 & 5 & 3 & 108 & 108 & 463 & 153 & 373 & 480\\
\hline
mushroom & 3858 & 7 & 4 & 1 & 19 & 16 & 230 & 59 & 58 & 58\\
\hline
btree & 6710 & 1103 &4 & 4 & 22 & 19 & 308 & 65 & 72 & 89\\
\hline
startup & 19312 & 621 &5 & 3 & 66 & 66 & 345 & 92 & 178 & 230\\
\hline
sunflow & 15615 & 85 & 5 & 2 & 30 & 27 & 315 & 43 & 102 & 124\\
\hline
compress & 16157 & 1483 & 5 & 3 & 32 & 29 & 338 & 50 & 105 & 131\\
\hline
parser & 7856 & 112 & 4 & 1 & 22 & 19 & 320 & 64 & 73 & 83\\
\hline
scimark & 16270 & 2027 &5 & 5 & 32 & 29 & 134 & 49 & 106 & 140\\
\hline
serial & 69999 & 468 & 8 & 3 & 130 & 130 & \footnotesize{MEM-OUT} & \footnotesize{MEM-OUT} & 3964 & 4314\\
%\hline
%\hline
%\textbf{Max} & &  & & &  \textbf{108} & \textbf{108} & \textbf{600} & \textbf{219} & \textbf{373} & \textbf{480}\\
\hline
\end{tabular}
\end{table}

%\vspace{em}
\section{Conclusion}
In this work we consider Dyck reachability problems for alias and 
data-dependence analysis.
For alias analysis, bidirected graphs are natural, for which we present
improved upper bounds, and present matching lower bounds to show our algorithm
is optimal.
For data-dependence analysis, we exploit constant treewidth property to present
almost-linear time algorithm.
We also show that for general graphs Dyck reachability bounds cannot be 
improved without achieving a major breakthrough.
%Interesting directions of future work would be to consider other optimizations
%along with our algorithms to develop more scalable tools for program analysis.

%% Acknowledgments
\begin{acks}                            %% acks environment is optional
The research was partly supported by Austrian Science Fund (FWF) Grant No P23499- N23, 
FWF NFN Grant No S11407-N23 (RiSE/SHiNE), and ERC Start grant (279307: Graph Games).
\end{acks}

%% Bibliography
\bibliography{bibliography}

\end{document}